\documentclass[a4paper,11pt]{article}

\usepackage{amsmath}    
\usepackage{amssymb}
\usepackage{amsthm} 
\usepackage{bbm}
\usepackage{booktabs}
\usepackage{microtype}
\usepackage{a4wide}
\usepackage{todonotes}

\usepackage[ruled,vlined]{algorithm2e}

\newtheorem{theorem}{Theorem}[section]

\newtheorem{lemma}[theorem]{Lemma}
\newtheorem{proposition}[theorem]{Proposition}

\theoremstyle{definition}
\newtheorem{definition}[theorem]{Definition}
\newtheorem{remark}[theorem]{Remark}

\newtheorem{assumption}{Assumption}

\usepackage{enumitem}
\usepackage{comment}
\usepackage[toc,page]{appendix}
\usepackage{color}
\usepackage{dsfont}
\usepackage{extarrows}
\usepackage{graphicx}   
\usepackage{mathrsfs}
\usepackage{verbatim}   
\usepackage{subfigure}  
\usepackage{hyperref}   

\newcommand{\Tr}{\mathrm{Tr}}

\title{\Large \bf
Measurement-Based Feedback of Open Quantum Systems: A Control-Theoretic Review and Tutorial
}
\date{ }

\author{Weichao Liang
\thanks{
{\small W. Liang is with the School of Automation Science and Engineering, Faculty of Electronic and Information Engineering, Xi’an Jiaotong University, Xi'an, Shaanxi, P.R. China (e-mail: weichao.liang@xjtu.edu.cn). }}
\and Jing Zhang
\thanks{
{\small J. Zhang is with the School of Automation Science and Engineering, Faculty of Electronic and Information Engineering, Xi’an Jiaotong University, Xi'an, Shaanxi, P.R. China (e-mail:zhangjing2022@xjtu.edu.cn). }}
\and Gaoyue Guo
\thanks{
{\small G. Guo is with the Laboratoire MICS and CNRS FR-3487, CentraleSupélec, Université Paris-Saclay, Gif-sur-Yvette, France (e-mail: gaoyue.guo@centralesupelec.fr). }}
\and Daoyi Dong
\thanks{
{\small D. Dong is with Australian Artificial Intelligence Institute, Faculty of  Engineering and Information Technology, University of Technology Sydney, Ultimo, Australia (e-mail: daoyidong@gmail.com). }}
}

\begin{document}

\maketitle

\begin{abstract}
This review develops a control-theoretic perspective on measurement-based feedback for continuously monitored open quantum systems, with the main analysis focused on finite-dimensional systems
governed by diffusive stochastic master equations. We introduce the relevant
state-space, invariant-subspace, and quantum non-demolition structures,
interpret quantum filtering as nonlinear observer dynamics, and review
open-loop asymptotics, filter stability, state-feedback stabilization,
robustness, and reduced-order observer-based control. Particular emphasis is
placed on a recurrence--contraction framework, in which Hamiltonian feedback
removes non-target invariant obstructions while measurement-induced dynamics
provide local exponential contraction. Although the detailed analysis is
developed for finite-dimensional diffusive models, the underlying
measurement--estimation--feedback architecture is relevant across a broad
range of quantum platforms.
We further discuss implementation challenges and open problems involving
scalable estimation, sampling and delay, adaptation, hybrid and
non-Markovian dynamics, practical stability, optimal control, and
learning-based design. By organizing these developments around
invariance, estimation, recurrence, and contraction, the review provides
a tutorial bridge between measurement-based quantum feedback and
nonlinear stochastic control.
\end{abstract}




\section{Introduction}

Quantum control provides a systematic framework for manipulating quantum
systems and plays an important role in quantum information processing,
quantum sensing, and quantum computing
~\cite{nielsen2010quantum,dong2010quantum,zhang2017quantum}.
In realistic devices, interactions with uncontrolled environments and
measurement apparatus make open-system descriptions unavoidable, while
model uncertainty, imperfect measurements, and actuator limitations further
constrain achievable performance. Feedback provides a natural mechanism for
using information acquired during the evolution to regulate the conditional
dynamics and stabilize prescribed quantum states or subspaces
~\cite{wiseman2009quantum,barchielli2009quantum,altafini2012modeling,
jacobs2014quantum,jordan2024quantum}.

Quantum feedback can be realized through different information structures.
In coherent feedback, the plant is interconnected directly with another
quantum dynamical system, without introducing an intermediate classical
measurement record or state estimate
~\cite{gough2009quantum,gough2009series,combes2017slh,nurdin2017linear}.
In measurement-based feedback, an output field is measured and the resulting
classical record is processed to determine the control action. One important
architecture is direct (or Markovian) feedback, in which the measured current
is fed back without explicitly reconstructing the conditional state
~\cite{wiseman1993quantum,wiseman1994quantum,wiseman2009quantum}.
A second architecture is filtering-based (or Bayesian) feedback, in which a
conditional state is propagated from the measurement record and used as the
information state for control. Early developments established continuous
state-estimation-based feedback and extended this architecture to tasks such
as continuous quantum error correction
~\cite{doherty1999feedback,doherty2000quantum,ahn2002continuous}.
Subsequent work developed state-based preparation and stochastic
stabilization methods for continuously monitored quantum systems
~\cite{stockton2004deterministic,van2005feedback,
mirrahimi2007stabilizing,tsumura2008global,qi2010measurement,qi2012further,
ticozzi2012stabilization,liu2016lyapunov,liang2019exponential}.
The core of this review concerns this filtering-based setting.

For filtering-based feedback, the conditional density operator serves as the
state variable relevant to control design. Under continuous diffusive
monitoring, conditioning on the measurement record yields a nonlinear
stochastic master equation (SME) whose trajectories remain in the compact
convex set of density operators. The stochasticity is not an externally
imposed disturbance added to a deterministic model; rather, it originates
from conditioning on a random measurement record, while the same
system--probe interaction that provides information also induces measurement
back-action. Thus, filtering-based quantum feedback couples nonlinear
quantum filtering with stochastic feedback design under the intrinsic
positivity and trace constraints of the quantum state. Its practical
realization further introduces limitations associated with real-time
estimation, computational complexity, finite detector and
controller bandwidth, feedback delay, and model uncertainty.

The present review focuses primarily on finite-dimensional open quantum
systems under continuous diffusive monitoring, including homodyne and
heterodyne detection. This setting already exhibits the principal structural
features relevant to stochastic-control analysis of quantum feedback:
nonlinear filtering, state-dependent multiplicative diffusion, measurement
back-action, invariant subspaces, and positivity and trace constraints on
the state. At the same time, finite-dimensional diffusive SMEs provide a
mathematically tractable framework in which invariance, recurrence,
stochastic stability, convergence rates, robustness, and observer-based
feedback can be studied rigorously. A substantial control-oriented literature
relevant to this program has been developed, ranging from invariant-subspace and
open-loop stability analysis to QND state reduction, feedback stabilization,
robustness to initialization and parameter mismatch, and reduced-order
filtering and control
~\cite{ticozzi2008quantum,benoist2014large,benoist2017exponential,
ticozzi2012stabilization,cardona2020exponential,liang2021GHZ,
liang2021robust,liang2024model,liang2025exploring,
amini2025feedback,liang2025reduced}.

However, the underlying measurement--estimation--feedback paradigm is not
tied to a particular physical platform. Measurement-based feedback has been
investigated and experimentally demonstrated in cavity-QED systems
~\cite{sayrin2011realtime}, superconducting circuits
~\cite{vijay2012stabilizing,riste2013deterministic}, trapped-ion systems
~\cite{bushev2006feedback}, collective atomic-spin ensembles
~\cite{inoue2013unconditional}, solid-state spin systems
~\cite{blok2014manipulating}, and mechanical or optomechanical platforms
~\cite{rossi2018measurement}. These realizations differ substantially in
their natural state spaces, measurement mechanisms, characteristic time
scales, and controller implementations. Depending on the platform and
measurement scheme, the appropriate conditional description may instead
involve discrete-time quantum trajectories, counting processes,
continuous-variable stochastic equations, or hybrid models. Therefore, the
detailed results reviewed below should not be interpreted as applying
without modification to every such realization. What is common across these
settings is the underlying information structure: measurement records provide
information about a conditional quantum state while simultaneously inducing
back-action, and this information is processed in real time to determine
subsequent actuation. In this sense, finite-dimensional diffusive SMEs provide
a useful benchmark in which fundamental problems of estimation, feedback,
stability, and robustness can be formulated and analyzed precisely.

Several monographs, surveys, and tutorial works provide broad treatments of
quantum measurement, filtering, feedback, and quantum control from
complementary physical and mathematical perspectives
~\cite{wiseman2009quantum,barchielli2009quantum,dong2010quantum,
altafini2012modeling,jacobs2014quantum,zhang2017quantum,
jordan2024quantum}.
The objective of the present review is not to provide an encyclopedic
survey of all quantum-feedback architectures, but to develop a
control-theoretic synthesis of the mechanisms by which continuous
measurement, estimation, and feedback generate closed-loop stochastic
stability. Accordingly, we regard the conditional density operator as an
information state and organize the discussion around the invariance and
attractivity of target subspaces, the asymptotic information extracted from
the measurement record, the stability of quantum filters under initialization
and model mismatch, the mechanisms responsible for global recurrence and
local contraction, and the effects of reduced-order estimation, uncertainty, and delay on closed-loop stability. This viewpoint allows results developed for different quantum-feedback models to be compared through their underlying control mechanisms rather than through particular physical realizations alone.

We first formulate continuously monitored quantum dynamics as
stochastic control systems and examine their open-loop asymptotic structure.
We then interpret quantum filtering as nonlinear observer dynamics before
turning to state-feedback and observer-based feedback stabilization.
Particular attention is given to quantum non-demolition (QND) models, for
which the measurement structure makes the relation among state reduction,
reduced-order filtering, and feedback stabilization especially transparent.
A recurring control-theoretic principle is the separation of global
recurrence from local contraction: feedback removes non-target invariant
obstructions, while continuous measurement supplies the local stochastic
contraction that determines the convergence rate. This separation between
global recurrence and local contraction provides a common control-theoretic
interpretation of stabilization results originally developed for different
model classes.

Section~2 formulates continuously monitored quantum dynamics as
state-constrained stochastic systems, Section~3 develops their open-loop
asymptotic structure, and Section~4 interprets quantum filtering as
nonlinear observer dynamics. Sections~5 and~6 then show how familiar
control-theoretic concepts---invariance, accessibility, support-theorem
arguments, stochastic Lyapunov methods, recurrence, and robustness---enter
the design and analysis of state-feedback and reduced-order
observer-based quantum controllers.
Sections~2--5.2 primarily review established results from the literature,
whereas Sections~5.3 and~6 provide a complementary tutorial synthesis.
Building on the filtering-based and observer-based stabilization results
developed in
\cite{liang2019exponential,liang2021GHZ,liang2021robust,
liang2024model,liang2025exploring,liang2025reduced},
we recast their assumptions and proof mechanisms within a common QND
framework according to four principal roles: target invariance, global
recurrence, local exponential contraction, and observer consistency and
robustness. Using transparent sufficient conditions, these sections expose
the shared control-theoretic structure underlying the different results
and clarify how recurrence, contraction, and robustness interact in
closed-loop stabilization. The final section examines how this framework
is affected by scalable estimation, finite-rate implementation,
uncertainty and adaptation, hybrid and non-Markovian dynamics,
non-invariant targets, stochastic optimal control, and learning-based
design.

\paragraph{Notation}
We denote by $\mathcal{B}(\mathcal{H})$ the space of all linear operators on a finite-dimensional Hilbert space $\mathcal{H}$, the adjoint $A\in\mathcal{B}(\mathcal{H})$ by $A^*$, and define $\mathcal{B}_{*}(\mathcal{H}):=\{X\in\mathcal{B}(\mathcal{H})|X=X^*\}$, $\mathcal{B}_{\geq 0}(\mathcal{H}):=\{X\in\mathcal{B}(\mathcal{H})|X\geq 0\}$ and $\mathcal{B}_{>0}(\mathcal{H}):=\{X\in\mathcal{B}(\mathcal{H})|X> 0\}$. 
We use $\mathbf{I}$ to denote the identity operator on $\mathcal{H}$, and $\mathds{1}$ for indicator functions. The imaginary unit is denoted by $i$. For any finite positive integer $n$, we define $[n]:=\{1,\dots,n\}$.
The commutator and anticommutator of  $A,B\in\mathcal{B}(\mathcal{H})$ are denoted by $[A,B]:=AB-BA$ and $\{A,B\}=AB+BA$, respectively.
$\Tr(A)$ denotes the trace of $A$. The trace norm of $A\in\mathcal{B}(\mathcal{H})$ is denoted by $\|A\|_1:=\mathrm{Tr}\sqrt{AA^*}$.  For $q\in\mathbb R^n$, its $\ell_1$ norm is denoted by $\|q\|_1$.
We denote by $\boldsymbol\Delta^{\circ}$ the relative interior of a simplex $\boldsymbol\Delta$. For $x\in\mathbb{C}$, $\Re\{x\}$ is the real part of $x$. 

\section{Continuously monitored systems and feedback information structures}
\label{sec:open-quantum-systems-sme}

This section introduces the finite-dimensional open quantum systems
considered throughout the review and establishes the notation used later
for SMEs, invariant subspaces, and feedback stabilization. The purpose is not to derive SMEs from quantum stochastic calculus, but to formulate them in a form suitable for stochastic-control analysis. For rigorous derivations, we refer the reader to~\cite{belavkin1989nondemolition,bouten2007introduction}. From this viewpoint, the state space is the compact convex set of density operators, the drift is a Lindblad generator, the observation process is a noisy measurement record, and the conditional state satisfies a nonlinear stochastic differential equation. We emphasize the structural properties that will be used later: invariance of the state space, target subspaces, the difference between unconditional and conditional dynamics, QND decompositions, and the information available to a feedback controller.

\subsection{Density operators and state space}
\label{subsec:density-operators-subspaces}

Let \(\mathcal H\simeq \mathbb C^N\) be a finite-dimensional Hilbert space. The state of a finite-dimensional quantum system is represented by a density operator, which is an element of
\begin{equation*}
\mathcal S(\mathcal H)
:=
\left\{
\rho\in\mathcal B(\mathcal H):
\rho=\rho^*,\ \rho\geq 0,\ \operatorname{Tr}(\rho)=1
\right\}.
\end{equation*}
The set \(\mathcal S(\mathcal H)\) is compact and convex. Its extreme points are precisely the rank-one orthogonal projectors $\rho=|\psi\rangle\langle\psi|$ with $|\psi\rangle\in\mathcal{H}$ satisfying $\langle\psi|\psi\rangle=1$,  which are called pure states. For \(\rho\in\mathcal S(\mathcal H)\), we have
\[
\rho \text{ is pure} \quad 
\Longleftrightarrow \quad 
\rho^2=\rho \quad 
\Longleftrightarrow \quad 
\operatorname{rank}(\rho)=1 \quad 
\Longleftrightarrow \quad 
\operatorname{Tr}(\rho^2)=1.
\]
A state that is not pure is called mixed. 
The distinction between pure and mixed states is fundamental in the analysis of open quantum systems. Hamiltonian dynamics preserve the spectrum of the density operator and therefore preserve purity. By contrast, interaction with an environment may change the spectrum of the reduced state, and may transform a pure state into a mixed state. Continuous observation affects the conditional state differently: under suitable observability or distinguishability conditions, the information acquired from the measurement record may lead to asymptotic purification or reduction toward a lower-dimensional invariant sector. Dissipation and measurement back-action provide non-unitary mechanisms, which are necessary for irreversible state preparation.

The compact convex geometry of \(\mathcal S(\mathcal H)\) is important for feedback stabilization; see~\cite{bengtsson2017geometry} for more details. In contrast with classical diffusions on Euclidean spaces, the SME evolves on a constrained state space, and many stabilization targets, especially pure states, lie on its boundary.

For the stabilization problems considered later, the target may be a subspace or a single state. Let \(\mathcal V\subseteq\mathcal H\) be a linear subspace and let \(\Pi_{\mathcal V}\) be the orthogonal projector onto \(\mathcal V\). We define the set of states supported on \(\mathcal V\) by
\begin{equation*}
\mathcal I(\mathcal V):=\left\{\rho\in\mathcal S(\mathcal H):\operatorname{supp}(\rho)\subseteq\mathcal V\right\}.
\end{equation*}
For a fixed \(\rho\in\mathcal S(\mathcal H)\), the following statements are equivalent:
\begin{equation*}
    \rho\in\mathcal I(\mathcal V) \quad \Longleftrightarrow \quad  \Pi_{\mathcal V}\rho\Pi_{\mathcal V}=\rho \quad  \Longleftrightarrow \quad  \Tr(\Pi_{\mathcal V}\rho)=1 \quad  \Longleftrightarrow \quad  \Tr(\Pi_{\mathcal V^\perp}\rho)=0.
\end{equation*}
If \(\dim\mathcal V=1\), \(\mathcal I(\mathcal V)\) consists of a single pure state. Higher-dimensional targets arise naturally in decoherence-free and encoding subspaces, including code spaces used in quantum error correction~\cite{lidar2013quantum}. The corresponding control objective is to render \(\mathcal I(\mathcal V)\) invariant and attractive.

\subsection{Hamiltonian dynamics and  the stabilization obstruction}
\label{subsec:closed-hamiltonian}

For an isolated finite-dimensional quantum system, a state vector \(|\psi_t\rangle\) evolves according to the Schrödinger equation. At the density operator level, the corresponding evolution is the Liouville--von Neumann equation
\begin{equation}
\dot{\rho}_t=-i[H_0,\rho_t],
\label{eq:von-neumann-free}
\end{equation}
where \(H_0=H_0^*\) is the free Hamiltonian. The solution is $\rho_t=U_t\rho_0U_t^*$ where the propagator satisfies $\dot U_t=-iH_0U_t$ with $U_0=\mathbf{I}.$
Thus, the evolution is unitary and preserves the spectrum of the density operator: $\operatorname{spec}(\rho_t)=\operatorname{spec}(\rho_0)$ for all $t\geq 0$.
In particular, the purity \(\operatorname{Tr}(\rho_t^2)\), and more generally every spectral invariant \(\operatorname{Tr}(\rho_t^k)\) is conserved.

External classical fields can often be modeled semiclassically through a controlled Hamiltonian $$H(u_t)=H_0+\sum_{j=1}^m u_{j,t}H_j$$ where $H_j=H_j^*$ and $u_{j,t}\in\mathbb R$; see~\cite{d2021introduction} for more details. The corresponding controlled Liouville--von Neumann equation is
\begin{equation}
\dot\rho_t
=
-i[H(u_t),\rho_t].
\label{eq:controlled-von-neumann}
\end{equation}
If \(u_t\) is prescribed as a deterministic function of time, independent of the state and of any measurement record, then \eqref{eq:controlled-von-neumann} is an open-loop coherent-control system. Such models are fundamental in pulse design, NMR, atomic physics, superconducting circuits, and quantum gate synthesis. The word \emph{coherent} emphasizes that the evolution remains unitary. Hence, it can rotate states and synthesize unitary transformations, but cannot change the spectrum of \(\rho_t\). Consequently, Hamiltonian control alone cannot purify an unknown mixed state, reduce entropy, or make all initial states in \(\mathcal S(\mathcal H)\) converge asymptotically to the same prescribed pure state or target subspace. This is the basic stabilization obstruction for purely coherent dynamics. Therefore, the feedback stabilization problems considered in this review require non-unitary mechanisms. Dissipation, continuous measurement, and measurement-based feedback provide precisely such mechanisms: they can create irreversible convergence. The next subsections introduce these mechanisms through Lindblad equations and stochastic master equations.

\subsection{Unconditional Markovian dynamics: Lindblad equations}
\label{subsec:lindblad}

The isospectral obstruction discussed above shows that purely Hamiltonian dynamics cannot generate irreversible convergence. This is not only a limitation from the control viewpoint, it also reflects the fact that perfectly isolated quantum systems are idealizations. In practice, a quantum system almost always interacts with uncontrolled environmental degrees of freedom, such as vacuum field modes, phonons, thermal reservoirs, surrounding spins, or measurement devices. Under the usual Markovian approximation, the reduced state of the system is described by a Lindblad equation~\cite{gorini1976completely,lindblad1976generators,alicki2007quantum}.

For \(L\in\mathcal B(\mathcal H)\), define the dissipator
\begin{equation*}
\mathcal D_L(\rho)
:=
L\rho L^*
-
\frac12 L^*L\rho
-
\frac12 \rho L^*L.
\end{equation*}
The operator \(L\) is usually called a Lindblad operator, noise operator, or jump operator. It encodes an effective irreversible channel through which the system exchanges information or energy with its environment. Depending on the physical setting, such a channel may describe spontaneous emission, photon loss, dephasing, thermal relaxation, or unobserved measurement back-action. A controlled Lindblad equation takes the form
\begin{equation}
\dot\rho_t
=
-i[H(u_t),\rho_t]
+
\sum_{\nu=1}^r \gamma_\nu\mathcal D_{L_\nu}(\rho_t),
\label{eq:lindblad-master-equation}
\end{equation}
where \(H(u_t)\) is the controlled Hamiltonian. The coefficient \(\gamma_\nu\geq 0\) is the corresponding coupling strength, decay rate, or decoherence rate. It determines the time scale on which the \(\nu\)-th channel affects the reduced dynamics of the system. It can also be absorbed into the
channel operator by replacing \(L_\nu\) with \(\sqrt{\gamma_\nu}L_\nu\).

For every prescribed measurable and locally bounded input \(u_t\),
equation~\eqref{eq:lindblad-master-equation} defines a completely positive
and trace-preserving evolution. In particular, if
\(\rho_0\in\mathcal S(\mathcal H)\), then
\(\rho_t\in\mathcal S(\mathcal H)\) for all \(t\geq0\).
It is obtained either when the environmental outputs are not monitored or
when the measurement outcomes are averaged out. The Hamiltonian term
generates unitary motion, whereas the dissipative terms may change the
spectrum of the density operator. Thus, Lindblad dynamics may alter purity,
suppress coherences, admit stationary states, and render invariant
subspaces attractive.
The presence of dissipative terms alone does not imply convergence to a
unique state or subspace. The asymptotic behavior depends on the algebraic
structure of the Hamiltonian and Lindblad operators. Conversely, suitably
engineered Markovian dynamics can provide mechanisms for state preparation,
subspace stabilization, and entanglement generation that are unavailable
under purely Hamiltonian control~\cite{ticozzi2008quantum,schirmer2010stabilizing,ticozzi2012hamiltonian}.

\subsection{Diffusive continuous monitoring and stochastic master equations}
\label{subsec:sme}

Suppose that part of the environment is continuously monitored through
diffusive measurements. The normalized state conditioned on the measurement
history satisfies a stochastic master equation, also called a quantum
filtering equation~\cite{belavkin1989nondemolition,
bouten2007introduction,barchielli2009quantum}.

For a single observation process \(Y_t\), let $\mathcal F_t^Y:=\sigma\{Y_s:0\leq s\leq t\}$ denote the corresponding observation filtration. The conditional state
\(\rho_t\) is adapted to \((\mathcal F_t^Y)_{t\geq 0}\) and determines the
conditional expectations of system observables given the measurement record
up to time \(t\).
For \(L\in\mathcal B(\mathcal H)\), define
\begin{equation*}
\mathcal G_L(\rho)
:=
L\rho+\rho L^*
-
\operatorname{Tr}\bigl((L+L^*)\rho\bigr)\rho.
\end{equation*}
The last term ensures normalization of the conditional state. In particular,
$
\operatorname{Tr}\bigl(\mathcal G_L(\rho)\bigr)=0.
$
The phase of the local oscillator in a homodyne measurement may be absorbed
into the definition of \(L\).
For a single monitored channel with coupling strength \(\gamma\geq 0\) and
detection efficiency \(\eta\in[0,1]\), the conditional state and the measurement output satisfy
\begin{align}
d\rho_t&=\left(-i[H(u_t),\rho_t]+\gamma\mathcal D_L(\rho_t)\right)dt+\sqrt{\eta\gamma}\mathcal G_L(\rho_t)dW_t,\label{eq:single-homodyne-sme}\\
dY_t&=\sqrt{\eta\gamma}\operatorname{Tr}\bigl((L+L^*)\rho_t\bigr)dt+dW_t,\label{eq:single-homodyne-output}
\end{align}
where the innovation process \(W_t\) is a standard Wiener process with respect to the observation filtration \((\mathcal F_t^Y)_{t\geq 0}\).
The drift of the observation process depends on the conditional
expectation of the measured quadrature \(L+L^*\), whereas the innovation
process determines the stochastic correction of the conditional state.
The same system--environment coupling generates the dissipative term
\(\gamma\mathcal D_L\), independently of whether the output is detected.
The detection efficiency determines the observed fraction of the output
information. If \(\eta=0\), the stochastic correction associated with the
channel vanishes, while the dissipative contribution
\(\gamma\mathcal D_L\) remains present. If \(\gamma=0\), the channel is
absent.

More generally, let \(C_k\), \(k\in[g]\), denote unobserved channels, and let
\(L_\nu\) with  \(\nu\in[r]\) denote monitored channels. Define the Lindblad
generator
\begin{equation*}
\mathcal L_u(\rho):=-i[H(u),\rho]+\sum_{k=1}^{g}\theta_k\mathcal D_{C_k}(\rho)+\sum_{\nu=1}^{r}\gamma_\nu\mathcal D_{L_\nu}(\rho),
\end{equation*}
where \(\theta_k\geq 0,\gamma_\nu> 0\). For each monitored channel, set $\iota_\nu:=\sqrt{\eta_\nu\gamma_\nu}>0$ with $\eta_\nu\in(0,1]$ and $\gamma_\nu> 0$. The conditional state and observation processes satisfy
\begin{align}
d\rho_t&=\mathcal L_{u_t}(\rho_t)dt+\sum_{\nu=1}^{r}\iota_\nu\mathcal G_{L_\nu}(\rho_t)dW_{\nu,t},\label{eq:multi-homodyne-sme}\\
dY_{\nu,t}&=\iota_\nu\operatorname{Tr}\bigl((L_\nu+L_\nu^*)\rho_t\bigr)dt+
dW_{\nu,t},\qquad\nu\in[r].\label{eq:multi-homodyne-output}
\end{align}
The joint observation filtration is $\mathcal F_t^Y:=\sigma\left\{Y_{\nu,s}:0\leq s\leq t,\ \nu\in[r]\right\}$, and $W_t=(W_{1,t},\ldots,W_{r,t})$ is an \(r\)-dimensional standard Wiener process with respect to this
filtration.
The term \(\theta_k\mathcal D_{C_k}\) represents the unconditional
contribution of the \(k\)-th unobserved channel. For the \(\nu\)-th
monitored channel, \(\gamma_\nu\mathcal D_{L_\nu}\) is its unconditional
dissipative contribution, whereas \(\iota_\nu\mathcal G_{L_\nu}\) determines the corresponding conditional measurement update. Heterodyne detection can be represented by two real observation channels associated with orthogonal field quadratures, with independent Wiener innovations and appropriately rescaled measurement operators~\cite{wiseman2009quantum}.

Under standard admissibility assumptions on \(u_t\), the finite-dimensional
SME~\eqref{eq:multi-homodyne-sme} admits a unique global strong solution.
In particular, this conclusion holds for the sufficiently regular feedback laws considered later. Moreover,  for any $\rho_0\in\mathcal S(\mathcal H)$, $\rho_t\in\mathcal S(\mathcal H)$ for all $t\geq 0$ almost surely; see, for instance, \cite[Section~3]{mirrahimi2007stabilizing} or \cite[Chapter~5]{barchielli2009quantum}.

\subsection{QND monitoring and invariant sectors}
\label{subsec:qnd}

The asymptotic behavior of a continuously monitored quantum system depends on the algebraic relation between the monitored operators and the uncontrolled dynamics. A principal setting considered in this review is quantum non-demolition (QND) monitoring~\cite{braginsky1980quantum,haroche2006exploring}. We use the term QND in the following finite-dimensional block-structured sense.

Let $$\mathcal H=\mathcal H_0\oplus\mathcal{H}_1\oplus \dots\oplus\mathcal{H}_d$$ be an orthogonal decomposition, and let $\Pi_j:=\Pi_{\mathcal H_j}$ denote the orthogonal projector onto \(\mathcal H_j\).
\begin{assumption}[QND structure]
\label{ass:QND}
$H_0=\mathrm{diag}[{H}_{0,0},\dots,{H}_{0,d}]$ with $ {H}_{0,j}\in  \mathcal{B}_{*}(\mathcal{H}_j)$,  $C_k=\mathrm{diag}[{C}_{k,0},\dots,{C}_{k,d}]$ with ${C}_{k,j}\in  \mathcal{B}(\mathcal{H}_j)$, and $L_\nu=\textstyle\sum^{d}_{j=0}l_{\nu,j}\Pi_{j}$ with $l_{\nu,j}\in\mathbb{C}$.
\end{assumption}

Under the QND structure and in the absence of a control Hamiltonian coupling distinct blocks, each set \(\mathcal I(\mathcal H_j)\) is invariant. Define the population of the \(j\)-th sector by 
\[ 
q_j(\rho) := \operatorname{Tr}(\Pi_j\rho), \qquad j=0,\ldots,d. 
\] 
For \(q=(q_0,\ldots,q_d)\), introduce 
\begin{equation}\label{eq:a_nuPsi}
a_\nu(q) := \sum_{m=0}^{d} \Re(l_{\nu,m})q_m, \qquad \Psi_\nu^j(q) := \Re(l_{\nu,j})-a_\nu(q). 
\end{equation}
Write $q_{j,t}=q_j(\rho_t)$, a direct computation from equations~\eqref{eq:multi-homodyne-sme}--\eqref{eq:multi-homodyne-output} yields
\begin{align}
dq_{j,t} &= 2q_{j,t} \sum_{\nu=1}^{r} \iota_\nu \Psi_\nu^j(q_t)dW_{\nu,t}, \qquad j=0,\ldots,d. \label{eq:qnd-population-dynamics} \\
dY_{\nu,t}&=2\iota_\nu\overline a_\nu(q_t)dt+dW_{\nu,t}.
\label{eq:qnd_reduced_output}
\end{align} 
Therefore, \(q_t\) is a bounded martingale and evolves in the probability simplex
\[
\boldsymbol\Delta_{d}:=\{q\in[0,1]^{d+1}:\textstyle\sum_{j=0}^{d}q_j=1\}.
\]
In particular, the unconditional QND dynamics preserve the sector populations, whereas the measurement record updates their conditional values through the martingale terms in equation~\eqref{eq:qnd-population-dynamics}. 
The distinguishability conditions governing asymptotic quantum-state reduction are introduced in the next section.

\subsection{Information structures for feedback control} \label{subsec:feedback-information-structure} 
The interpretation of the control input \(u_t\) depends on the information available to the controller. This distinction also determines whether the expectation of the conditional state satisfies a closed deterministic equation. If \(u_t\) is a prescribed deterministic function of time, then the control is open loop. Since \(\mathcal L_{u_t}\) is linear in \(\rho\) for every fixed \(u_t\), and the stochastic integrals in equation~\eqref{eq:multi-homodyne-sme} have zero expectation, the averaged state \(\overline\rho_t := \mathbb E[\rho_t] \) satisfies ${d \overline\rho_t}/{dt} = \mathcal L_{u_t}(\overline\rho_t).$ 
Thus, under deterministic open-loop control, the unconditional state is the expectation of the conditional state. 

A measurement-based control is a nonanticipating functional of the observation history. The present review focuses on filtering-based feedback. In the ideal conditional-state feedback formulation, $u_t = u(\rho_t)$, where \(\rho_t\) is the correctly initialized conditional state generated from the measurement record. This architecture is commonly referred to as Bayesian feedback in the physics literature. In control-theoretic terms, \(\rho_t\) is the information state on which the feedback law is evaluated. It is reconstructed from the observation record rather than measured directly. 

For state feedback $ u(\rho_t)$,  In general, \( \mathbb E[\mathcal L_{u(\rho_t)}(\rho_t)] \) cannot be expressed as a function of \(\mathbb E[\rho_t]\) alone. Therefore, closed-loop stability must be analyzed directly for the conditional stochastic dynamics. The ideal formulation $u(\rho_t)$ assumes that the initial state and all model parameters entering the quantum filter are known. If the filter is initialized from a nominal prior or constructed with nominal parameters, the controller propagates an estimated conditional state \(\hat\rho_t\), driven by the same measurement record. The resulting observer-based feedback law is $u_t = u(\hat{\rho}_t)$. The actual conditional state \(\rho_t\) and the estimated state \(\hat\rho_t\) are components of a coupled stochastic system driven by the common observation process. 

For a system satisfying the QND structure in Assumption~\ref{ass:QND}, the asymptotic sector information is encoded by the population vector \(q(\rho_t) \).  Accordingly, a reduced observer may propagate only an estimated population vector \( \hat q_t = ( \hat q_{0,t},\dots,\hat q_{d,t} ) \) and apply the reduced-order feedback law $u_t = u(\hat q_t)$. The distinction between the exact reduced filter associated with the uncontrolled QND model and the feedback-augmented reduced observer used in the closed-loop design will be specified in the corresponding observer-based stabilization section.  

Filtering-based feedback should be distinguished from Markovian feedback~\cite{wiseman1993quantum,wiseman1994quantum,wiseman2009quantum}. In the latter architecture, the measured current is fed directly into the Hamiltonian or dissipative generator, typically under an idealized negligible-delay approximation, without first reconstructing a conditional density operator. It should also be distinguished from coherent feedback~\cite{combes2017slh}, in which the controller is itself a quantum system and the interconnection does not involve an intermediate classical measurement record. These architectures are outside the principal scope of the present review. The information structure considered in the remainder of the paper is summarized by
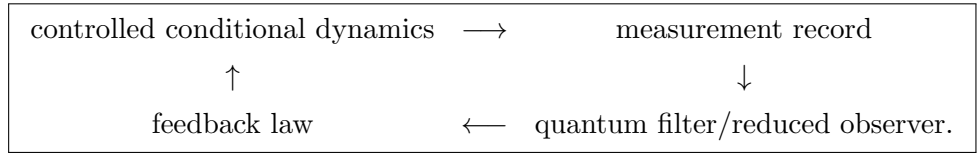
\begin{equation*}
\boxed{
\begin{array}{ccc}
\text{controlled conditional dynamics}
&
\longrightarrow
&
\text{measurement record}
\\[0.4em]
\uparrow
&&
\downarrow
\\[0.4em]
\text{feedback law}
&
\longleftarrow
&
\text{quantum filter/reduced observer}.
\end{array}
}
\end{equation*}
The subsequent sections examine, respectively, the open-loop asymptotic structure, stability of the quantum filter, conditional-state feedback stabilization, and reduced observer-based stabilization under model uncertainty.

\section{Open-loop asymptotics of stochastic master equations}
\label{sec:open_loop_asymptotics}

This section reviews the open-loop asymptotic properties of the diffusive
SME~\eqref{eq:multi-homodyne-sme}. Throughout the section, the control
input \(u_t\) is deterministic and sufficiently regular. The averaged state satisfies
$\frac{d}{dt}\bar{\rho}_t=\mathcal L_{u_t}(\bar{\rho}_t).$
Thus, under deterministic open-loop control, invariant and attractive
subspaces of the averaged Lindblad dynamics provide structural information
about the corresponding conditional trajectories. We first recall the
general invariant-subspace criteria and the relation between mean and
almost-sure stability. The QND specialization and measurement-induced state
reduction are considered subsequently.

\subsection{Invariant subspaces and attractivity}
\label{subsec:target_subspaces_distances}

Let \(\mathcal H_S\subsetneq\mathcal H\) be a nonzero target subspace and
write $\mathcal H = \mathcal H_S\oplus\mathcal H_R$ with $\mathcal H_R := \mathcal H_S^\perp$. 
\begin{definition}
The target subspace \(\mathcal H_S\) is said to be:
\begin{enumerate}
    \item \emph{invariant in mean} if $\rho_0\in\mathcal I(\mathcal H_S)$ then $\mathbb E(\rho_t)\in\mathcal I(\mathcal H_S)$ for all $t\geq 0$;
    \item \emph{invariant almost surely} if $\rho_0\in\mathcal I(\mathcal H_S)$ then $\rho_t\in\mathcal I(\mathcal H_S)$ for all $t\geq 0$ almost surely.
\end{enumerate}
\end{definition}
Every operator \(X\in\mathcal B(\mathcal H)\) admits the block
decomposition
\[
X
=
\begin{pmatrix}
X_S & X_P\\
X_Q & X_R
\end{pmatrix}
\]
with respect to \(\mathcal H=\mathcal H_S\oplus\mathcal H_R\), where \(X_S:\mathcal H_S\to\mathcal H_S\), \(X_R:\mathcal H_R\to\mathcal H_R\),
\(X_P:\mathcal H_R\to\mathcal H_S\), and \(X_Q:\mathcal H_S\to\mathcal H_R\).
The following criterion characterizes invariance directly from the
coefficients of the SME; see
 \cite{ticozzi2008quantum, baumgartner2008analysis,baumgartner2012structures}.

\begin{lemma}[Invariance criterion]
\label{lem:open_loop_invariance}
Consider the SME~\eqref{eq:multi-homodyne-sme} with deterministic input
\(u_t\). The subspace \(\mathcal H_S\) is invariant in mean if and only if
it is invariant almost surely. These equivalent properties hold if and only if
$C_{k,Q}=0$ and $L_{\nu,Q}=0$ for all $k\in[g]$ and $\nu\in[r]$, and
\begin{equation}
iH_P(u_t) = \frac{1}{2}\left(\sum_{k=1}^{g}\theta_k C_{k,S}^*C_{k,P}+\sum_{\nu=1}^{r}\gamma_\nu L_{\nu,S}^*L_{\nu,P}\right),\quad t\geq0.
\label{eq:open-loop-invariance-condition}
\end{equation}
\end{lemma}

Invariance is necessary for exact asymptotic stabilization: if
\(\mathcal I(\mathcal H_S)\) is not invariant, a trajectory initialized in
the target set may leave it immediately.
For comparison with norm-based stability notions, define the subspace error
\begin{equation*}
d_S(\rho):=\|\rho-\Pi_{\mathcal{H}_S}\rho\Pi_{\mathcal{H}_S}\|_1.
\end{equation*}
Then, \(d_S(\rho)=0\) if and only if \(\rho\in\mathcal I(\mathcal H_S)\).  
We now introduce the stability notions used in this paper. They extend the classical concepts of stochastic stability~\cite{mao2007stochastic,khasminskii2011stochastic} to density matrices and target subspaces. Motivated by~\cite{ticozzi2008quantum,benoist2017exponential}, all convergence properties will be formulated in terms of  $d_S(\rho)$.

\begin{definition}
The invariant subspace $\mathcal H_S\subset\mathcal H$  or the related subset $\mathcal{I}(\mathcal{H}_S)$ is said to be
\begin{enumerate}
    \item
    \emph{stable for the averaged dynamics}, if for every $\varepsilon \in (0,1)$, there exists $\delta = \delta(\varepsilon)>0$ such that $d_S(\bar{\rho_t})\leq  \varepsilon$ for all $t\geq 0$  whenever $d_S(\rho_0)<\delta$. 
    
    \item \emph{globally asymptotically stable (GAS) for the averaged dynamics}, if it is stable for the averaged dynamics and $\lim_{t\rightarrow\infty} d_S(\bar{\rho_t})= 0$ for all $\rho_0\in\mathcal{S}(\mathcal{H})$.
    
    \item
    \emph{globally exponentially stable (GES) for the averaged dynamics}, if there exists a pair of  constants $\lambda, c>0$ such that $d_S(\bar \rho_t)\leq c d_S(\rho_0) e^{-\lambda t}$ for all $t\geq0$ and $\rho_0\in\mathcal{S}(\mathcal{H}).$
    
    \item \emph{stable in probability} if, for every $\varepsilon\in(0,1)$ and $r>0$, there exists $\delta=\delta(\varepsilon,r)>0$ such that $ \mathbb P\big(\textstyle\sup_{t\ge0}d_S(\rho_t)<r\big)\ge 1-\varepsilon$ whenever $d_S(\rho_0)<\delta$;

    \item \emph{globally asymptotically stable (GAS) almost surely} if it is stable in probability and $\mathbb P\big(\lim_{t\to\infty} d_S(\rho_t)=0\big)=1$ for all $\rho_0\in\mathcal{S}(\mathcal{H}).$

    \item \emph{globally exponentially stable (GES) almost surely} if $\mathbb P\big(\textstyle\limsup_{t\to\infty}\tfrac{1}{t}\log  d_S(\rho_t)<0\big)=1$ for all $\rho_0\in \mathcal S(\mathcal H)\setminus\mathcal I(\mathcal H_S)$.
\end{enumerate}
\end{definition}

We now restrict to an autonomous open-loop input $u_t\equiv u.$
Under the invariance conditions of
Lemma~\ref{lem:open_loop_invariance}, the \(R\)-block of the averaged state
evolves autonomously, whose generator is given by
\begin{equation*}
\mathcal L_R(\sigma)
=-i[H_R(u),\sigma]+\sum_{k=1}^{g}\theta_k\mathfrak D_{C_k}^R(\sigma)
+\sum_{\nu=1}^{r}\gamma_\nu\mathfrak D_{L_\nu}^R(\sigma),
\end{equation*}
where
$\mathfrak D_A^R(\sigma):=A_R\sigma A_R^*-\frac{1}{2}\{A_R^*A_R+A_P^*A_P,\sigma\}.$
For every \(\sigma\geq0\),
\[
\operatorname{Tr}
\bigl(
\mathfrak D_A^R(\sigma)
\bigr)
=
-
\operatorname{Tr}
\bigl(
A_P^*A_P\sigma
\bigr)
\leq0.
\]
Hence, \(\mathcal L_R\) generates a positive trace-nonincreasing
semigroup on \(\mathcal B(\mathcal H_R)\). Moreover,
\[
\bar{\rho}_R(t)
=
e^{t\mathcal L_R}\rho_R(0).
\]
Define the spectral decay rate
$$\alpha_R:=-\max\{\Re z:z\in\operatorname{spec}(\mathcal L_R)\}.$$
The following result collects the geometric, spectral, Lyapunov, and
trajectory-level characterizations of open-loop attractivity; see
\cite{ticozzi2008quantum,benoist2017exponential}.

\begin{theorem}[Equivalent characterizations of open-loop attractivity]
\label{thm:open-loop-attractivity-characterization}
Assume that \(u_t\equiv u\) and that \(\mathcal H_S\) is invariant. Then
the following statements are equivalent:
\begin{enumerate}
    \item
    \(\mathcal H_S\) is globally asymptotically stable for the averaged
dynamics;
    \item
    \(\mathcal H_S\) is globally asymptotically stable almost surely for the conditional dynamics;

    \item
    there is no nonzero subspace
    \(\mathcal W\subseteq\mathcal H_R\) such that
    \(\mathcal I(\mathcal W)\) is invariant for the averaged dynamics and
    \(
    \mathcal W  \subseteq  \bigcap_{k=1}^{g}\ker C_{k,P} \cap \bigcap_{\nu=1}^{r}\ker L_{\nu,P};
    \)
    \item
    \(\alpha_R>0\);

    \item
    there exist \(K_R\in\mathcal B_{>0}(\mathcal H_R)\) and \(c>0\)
    such that
   $
    \mathcal L_R^*(K_R) \leq -cK_R,
    $
    where \(\mathcal L_R^*\) is the adjoint of \(\mathcal L_R\) with respect to the Hilbert--Schmidt inner product.
\end{enumerate}
\end{theorem}
Condition~(3) excludes invariant components supported entirely in the complement of the target. Condition~(4) characterizes decay of the complementary semigroup, while condition~(5) provides the converse linear Lyapunov function \(\Tr(K\rho)\) where $K=\mathrm{diag}(0,K_R)\in \mathcal{B}(\mathcal{H})$.
The rate \(\alpha_R\) is determined by the averaged Lindblad dynamics and
does not need to coincide with the almost-sure decay rate of the conditional
trajectories. Different monitoring schemes may have the same Lindblad
generator, and the same value of \(\alpha_R\), while producing
different sample-path Lyapunov exponents. Therefore, continuous observation may accelerate almost-sure convergence without changing the averaged
dynamics.

\subsection{QND systems and quantum state reduction}
\label{subsec:qnd_state_reduction_open_loop}

We now specialize the open-loop analysis to the QND structure in Assumption~\ref{ass:QND}. The sector populations satisfy equation~\eqref{eq:qnd-population-dynamics} and are bounded martingales. Then, continuous observation may identify one of the invariant sectors, although the selected sector is random.

The following condition ensures that distinct sectors generate
distinguishable measurement signals.

\begin{assumption}[Measurement distinguishability]
\label{ass:measurement-distinguishability}
For every \(i\neq j\), $\sum_{\nu=1}^{r}\iota_\nu^2[\Re(l_{\nu,i})-\Re(l_{\nu,j})]^2>0.$
\end{assumption}
Define the measurement-separation rate
\begin{equation*}
\mathfrak E_l:=\min_{i\neq j}\sum_{\nu=1}^{r}\iota_\nu^2[\Re(l_{\nu,i})-\Re(l_{\nu,j})]^2.
\end{equation*}
Under the measurement-distinguishability condition, one has \(\mathfrak E_l>0\).
Consider
\begin{equation*}
V_{\mathrm{QND}}(\rho)
:=
\sum_{0\leq i<j\leq d}
\sqrt{q_i(\rho)q_j(\rho)}\geq 0,
\end{equation*}
which is a symmetric measure of population dispersion. It vanishes exactly in $\bigcup_{j=0}^{d}\mathcal I(\mathcal H_j)$. The following result, adapted from~\cite[Theorem 2.5]{liang2024model} and reformulated in the notation used here, quantifies both the exponential decay of this dispersion and the asymptotic selection of a single QND sector.
\begin{theorem}[{\cite[Theorem 2.5]{liang2024model}}]
\label{Thm:QSR}
Assume that \(u_t\equiv0\), and that  Assumptions~\ref{ass:QND} and \ref{ass:measurement-distinguishability} hold. For every \(\rho_0\in\mathcal S(\mathcal H)\), 
\begin{align}
&\mathbb E\left[V_{\mathrm{QND}}(\rho_t)\right]\leq e^{-\mathfrak E_l t/2}V_{\mathrm{QND}}(\rho_0),\quad t\geq0, \label{eq:qnd-reduction-mean-rate}\\
&\limsup_{t\to\infty}\frac{1}{t}\log V_{\mathrm{QND}}(\rho_t)\leq-\frac{\mathfrak E_l}{2}\quad\text{a.s.}\label{eq:qnd-reduction-as-rate}
\end{align}
Moreover, there exists a random variable
\(\mathcal R\in\{0,\ldots,d\}\) such that
\begin{align}
&\lim_{t\to\infty}q_j(\rho_t)=\mathds 1_{\{R=j\}},\quad\text{a.s.},
\label{eq:qnd-random-sector-selection}\\
&\mathbb P(\mathcal R=j)=q_j(\rho_0)=\Tr(\Pi_j\rho_0), \quad
j=0,\dots,d.
\label{eq:qnd-sector-selection-probability}
\end{align}
\end{theorem}

\begin{proof}[Proof]
Let \( I_+ := \left\{ j\in\{0,\ldots,d\}: q_j(\rho_0)>0 \right\}. \) If \(j\notin I_+\), then equation~\eqref{eq:qnd-population-dynamics} implies that \(q_j(\rho_t)=0\) for all \(t\geq0\). For \(j\in I_+\), the multiplicative form of~\eqref{eq:qnd-population-dynamics} yields $q_j(\rho_t)>0$ for all $t\geq 0$ almost surely. Hence, the face of the probability simplex determined by \(I_+\) is invariant, and \(V_{\mathrm{QND}}\) is twice continuously differentiable in its relative interior. A direct computation of the infinitesimal generator on the invariant face determined by \(I_+\) yields
\begin{equation*}
\mathscr L V_{\mathrm{QND}}(\rho)
=-\frac{1}{2}\sum_{0\leq i<j\leq d,\ i,j\in I_+ }\left[\sum_{\nu}\iota_\nu^2\big(\Re(l_{\nu,i})-\Re(l_{\nu,j})\big)^2\right]\sqrt{q_i(\rho)q_j(\rho)}\leq
-\frac{\mathfrak E_l}{2}V_{\mathrm{QND}}(\rho).
\end{equation*}
Dynkin's formula and Grönwall's inequality yield
\eqref{eq:qnd-reduction-mean-rate}, while the standard stochastic
Lyapunov theorem
\cite{mao2007stochastic,khasminskii2011stochastic}
implies~\eqref{eq:qnd-reduction-as-rate}.

Since the populations are bounded martingales, they converge almost surely by Doob’s martingale convergence theorem.
The convergence of \(V_{\mathrm{QND}}(\rho_t)\) toward zero, together with
\(\sum_{j=0}^{d}q_j(\rho_t)=1\), implies that the limiting vector is a
vertex of the probability simplex, which proves
\eqref{eq:qnd-random-sector-selection}. By the martingale property, we have $\mathbb E[q_j(\rho_t)]=q_j(\rho_0)$ for all $t\geq 0$, which yields~\eqref{eq:qnd-sector-selection-probability}.
\end{proof}
Thus, open-loop QND monitoring produces random sector selection with
probabilities determined by the initial populations. It does not make a
prescribed sector globally attractive. Measurement-based feedback will
subsequently be used to replace this random reduction mechanism by
stabilization of a selected target state or subspace.

Sharper asymptotic exponents can be obtained in the nondegenerate QND setting by combining martingale methods with a Girsanov transformation; see~\cite{benoist2014large}. The Lyapunov estimate above is generally more
conservative, but it extends naturally to block-valued QND sectors and is consistent with the stability methods used later in the review.

\subsection{Robustness of open-loop attractivity}
\label{subsec:ISS_open_loop}

The preceding results concern an exact model for which the target set
\(\mathcal I(\mathcal H_S)\) is invariant. Model perturbations lead to two
distinct robustness questions. If the perturbation preserves invariance, the question is whether attractivity is retained. If invariance is destroyed,
exact convergence cannot generally be expected, and the appropriate
objective is instead a practical bound on $d_S(\rho)$.

We restrict to a constant open-loop input \(u_t\equiv u\). Consider the
perturbed operators $\bar H(u)=H(u)+\alpha\widetilde H$, $\bar L_\nu=\sqrt{\gamma_\nu}L_\nu+\beta\widetilde L_\nu$ and $\bar C_k=\sqrt{\theta_k}C_k+\zeta\widetilde C_k$ where \(\alpha,\beta,\zeta\in\mathbb R\). We also allow additional
unmodeled channels with operators \(A_j\in\mathcal B(\mathcal H)\) and
common intensity \(\kappa\geq0\). The perturbed SME is
\begin{align}
d\sigma_t=\big(\mathcal L_u(\sigma_t)+\mathcal F_{\xi}(\sigma_t)\big)dt+
\sum_{\nu=1}^{r}
\sqrt{\eta_\nu}
\mathcal G_{\bar L_\nu}(\sigma_t)dW_{\nu,t},
\label{eq:perturbed_sme_ISS}
\end{align}
where the perturbation generator is
\begin{align*}
\mathcal F_{\xi}(\rho)
:=&-i[\alpha\widetilde H,\rho]
+
\sum_{\nu=1}^{r}
\big(\mathcal D_{\bar L_{\nu}}(\rho)-\mathcal D_{\sqrt{\gamma_\nu} L_\nu}(\rho)
\big)\\
&+\sum_{k=1}^{g}
\big(\mathcal D_{\bar C_k}(\rho)-\mathcal D_{\sqrt{\theta_k} C_k}(\rho)\big)
+\kappa\sum_{j=1}^{p}\mathcal D_{A_j}(\rho),
\end{align*}
with $\xi:= (\alpha,\beta,\zeta,\kappa)$.
Assume that \(\mathcal H_S\) is invariant for the nominal dynamics. Applying
Lemma~\ref{lem:open_loop_invariance} to the perturbed operators gives a
parameter-dependent invariance condition. For robustness, it is more useful to
impose a stronger condition that is uniform in the perturbation amplitudes.
\begin{assumption}[Robust invariance]
\label{ass:robust-invariance}
$\forall k\in[g]$, $\forall \nu \in[r]$, $\forall j\in[p]$, $\tilde{L}_{\nu,Q}=\tilde{C}_{k,Q}=A_{j,Q}=0$, $\tilde{H}_{P}=0$, and
    \begin{align*}
        & \textstyle\sum_{\nu=1}^r \big(L^*_{\nu,S} \tilde{L}_{\nu,P} + \tilde{L}^*_{\nu,S} L_{\nu,P}\big) = 0,\quad \textstyle\sum_{\nu=1}^r \big(\tilde{L}^*_{\nu,S} \tilde{L}_{\nu,P}\big) = 0, \\
        & \textstyle\sum_{k=1}^g \big(C^*_{k,S} \tilde{C}_{k,P} + \tilde{C}^*_{k,S} C_{k,P}\big) = 0,\quad \textstyle\sum_{k=1}^g\big(\tilde{C}^*_{k,S} \tilde{C}_{k,P}\big) = 0,\quad  \textstyle\sum_{j=1}^p A^*_{j,S} A_{j,P} = 0.
    \end{align*}
\end{assumption}
Thus, if the nominal dynamics leave \(\mathcal H_S\) invariant and
the robust-invariance condition holds, then \(\mathcal H_S\) remains invariant for
equation~\eqref{eq:perturbed_sme_ISS} for every value of
$\xi$.
The robust QND models considered in Section~\ref{subsec:robust-observer-feedback} belong to this structure-preserving regime: the observed operators retain their QND block structure, while the remaining uncontrolled dynamics continues to satisfy the target-invariance conditions. Hence, the uncertainties studied there modify the measurement information and the observer dynamics without inducing leakage from the target subspace.
\begin{proposition}[{\cite[Proposition 3.3]{liang2025exploring}}]
\label{prop:genericgas}
Assume that \(\alpha_R>0\) and  the robust-invariance condition in Assumption~\ref{ass:robust-invariance} holds. Then, there exists a neighborhood \(\mathcal U\) of the origin such that \(\mathcal H_S\) remains GES for the averaged dynamics and almost surely for every \(\xi\in\mathcal U\). Moreover, the same property holds for Lebesgue-almost every value of \(\xi\).
\end{proposition}
The local statement follows from continuity of the spectrum of the finite-dimensional complementary generator. The generic statement is based on the dissipation-induced decomposition and excludes only an exceptional parameter set of Lebesgue measure zero; see \cite{ticozzi2012hamiltonian,liang2025exploring}.

We next consider perturbations that do not necessarily satisfy the robust-invariance condition. In this case, exact convergence to \(\mathcal I(\mathcal H_S)\) cannot generally be expected, since the perturbed dynamics may transfer population out of the target subspace. Nevertheless, the nominal converse Lyapunov operator yields a practical robustness estimate.
There exists a constant \(c_d>0\) such that $d_S(\rho)^2\leq c_d\Tr(\Pi_{\mathcal H_R}\rho)$ for all $\rho\in\mathcal S(\mathcal H)$. 
\begin{proposition}[{\cite[Proposition 3.4]{liang2025exploring}}]
\label{Prop:ISS}
Assume that \(\alpha_R>0\), and let
\(K_R\in\mathcal B_{>0}(\mathcal H_R)\) and \(c>0\) satisfy
\(\mathcal L_R^*(K_R)\leq-cK_R.\) Set $K:=\operatorname{diag}(0,K_R)$ and
$D_\xi:=\sup_{\rho\in\mathcal S(\mathcal H)}\max\{0,\Tr(
K\mathcal F_{\xi}(\rho))\}$.
Then, for every initial state
\(\sigma\in\mathcal S(\mathcal H)\), the solution of
equation~\eqref{eq:perturbed_sme_ISS} with \(\sigma_0=\sigma\) satisfies
\begin{equation}
\mathbb E\bigl[d_S(\sigma_t)\bigr]\leq B_\xi(t,\sigma),\quad t\geq0,
\label{Eq:boundedness}
\end{equation}
where
$
B_\xi(t,\sigma) := [(\Tr(K\sigma)e^{-ct}+(1-e^{-ct}){D_\xi}/{c}){c_d}/{\lambda_{\min}(K_R)}]^{1/2}.
$
Moreover, for every \(\delta\geq1\),
\begin{equation}
\mathbb P \left( d_S(\sigma_t) \leq \delta B_\xi(t,\sigma) \right) \geq 1-{1}/{\delta}, \quad t\geq0.
\label{Eq:ISS}
\end{equation}
\end{proposition}
Indeed, by It\^o's formula~\cite{ikeda2014stochastic}, we have  \( \frac{d}{dt} \mathbb E [ \Tr(K\sigma_t) ] \leq -c \mathbb E [ \Tr(K\sigma_t) ] + D_\xi. \) 
Combining the resulting Gr\"onwall estimate and Jensen's inequality yields~\eqref{Eq:boundedness}, while~\eqref{Eq:ISS} follows from Markov's inequality. In particular, \[ \limsup_{t\to\infty} \mathbb E \left[ d_S(\sigma_t) \right] \leq \left( \frac{c_dD_\xi} {c\lambda_{\min}(K_R)} \right)^{1/2}. \] Thus, perturbations that destroy invariance replace exact attractivity by a mean practical-stability estimate for the conditional trajectories, whose size is determined by the perturbation magnitude \(D_\xi\) and the nominal stability margin~\(c\).

\section{Quantum filtering as observer dynamics}
\label{sec:observer-sme}

In the preceding sections, \(\rho_t\) denotes the conditional state generated
from the true initial state and model parameters. In practice, the controller
propagates an estimated conditional state \(\hat\rho_t\) from the same
measurement record, using a possibly incorrect initial state and misspecified model parameters. From a control-theoretic perspective,
\(\hat\rho_t\) is a nonlinear observer state.
First, we consider initialization mismatch under an exact model and review
observable-space and fidelity-based filter-stability results. Then, we specialize to QND systems, for which reduced filters subject to both initialization and parameter mismatch admit explicit exponential convergence estimates.

\subsection{Filter stability under initialization mismatch}
\label{subsec:filter_stability_observability}

Given an estimated initial state
\(\hat\rho_0\in\mathcal S(\mathcal H)\), define the estimated innovation
processes by
\begin{equation*}
d\widehat W_{\nu,t}:=dY_{\nu,t}-\iota_\nu\Tr\big((L_\nu+L_\nu^*)\hat\rho_t\bigr)dt,
\quad
\nu\in[r].
\end{equation*}
The estimated filter satisfies
\begin{equation}
d\hat\rho_t=\mathcal L_{u_t}(\hat\rho_t)dt+\sum_{\nu=1}^{r}\iota_\nu\mathcal G_{L_\nu}(\hat\rho_t)d\widehat W_{\nu,t}.
\label{eq:estimated-filter}
\end{equation}
The true and estimated filters are driven by the same measurement
record but generally have different innovation processes. Indeed,
equation~\eqref{eq:multi-homodyne-output} implies
\begin{equation*}
d\widehat W_{\nu,t}=dW_{\nu,t}+\iota_\nu\Tr\bigl((L_\nu+L_\nu^*)(\rho_t-\hat\rho_t)\bigr)dt.
\end{equation*}
Thus, under the probability law generated by the true initial state,
\(\widehat W_{\nu,t}\) is not generally a Wiener process, and
\((\rho_t,\hat\rho_t)\) forms a coupled stochastic system.
In finite dimension, \(\operatorname{supp}(\rho_0)\subseteq\operatorname{supp}(\hat\rho_0)\) is a sufficient condition for absolute continuity of the corresponding observation laws. In particular, this condition holds for every \(\rho_0\in\mathcal S(\mathcal H)\) whenever \(\hat\rho_0>0\). 

We now restrict to \(u_t\equiv u\). The adjoint Lindblad generator is
\begin{equation*}
\mathcal L_u^*(X)=i[H(u),X]+\sum_{k=1}^{g}\theta_k\mathcal D_{C_k}^*(X)+\sum_{\nu=1}^{r}\gamma_\nu\mathcal D_{L_\nu}^*(X),
\end{equation*}
where $\mathcal D_A^*(X):=A^*XA-\frac{1}{2}\left\{A^*A,X\right\}.$ Define
$
\mathcal K_\nu(X):=L_\nu^*X+XL_\nu.
$
Set $\mathcal O_0:=\operatorname{span}\{\mathbf I\}$
and define recursively
\begin{equation*}
\mathcal O_{n+1}:=\operatorname{span}\big(\mathcal O_n\cup\mathcal L_u^*(\mathcal O_n)\cup\textstyle\bigcup_{\nu=1}^{r}\mathcal K_\nu(\mathcal O_n)\big).
\end{equation*}
Since \(\mathcal B_*(\mathcal H)\) is finite dimensional, the sequence \((\mathcal O_n)_{n\geq0}\) stabilizes after finitely many iterations. Its terminal value, denoted by \(\mathcal O\), is the observable space. The model is observable when \( \mathcal O = \mathcal B_*(\mathcal H). \)

The following result is a finite-dimensional formulation of the quantum filter stability theorem.
\begin{theorem}[{\cite[Theorem 2.1]{handel2009stability}; \cite[Theorem 5.3.12]{van2007filtering}}]
\label{thm:quantum_filter_stability}
Assume that \(u_t\equiv u\) and \(\operatorname{supp}(\rho_0)\subseteq\operatorname{supp}(\hat\rho_0)\). Then, for every \(X\in\mathcal O\),
\begin{equation*}
\textstyle \lim_{t\to\infty}\mathbb E_{\rho_0}[|\Tr(X(\rho_t-\hat\rho_t))|]=0.
\end{equation*}
In particular, for every observed channel \(\nu\),
\begin{equation}
\textstyle \lim_{t\to\infty}\mathbb E_{\rho_0}[|\Tr((L_\nu+L_\nu^*)(\rho_t-\hat\rho_t))|]=0.
\label{eq:output_signal_filter_stability}
\end{equation}
If, in addition,
\(
\mathcal O=\mathcal B_*(\mathcal H),
\)
then $\lim_{t\to\infty}\mathbb E_{\rho_0}[\|\rho_t-\hat\rho_t\|_1]=0.$
\end{theorem}
The output convergence~\eqref{eq:output_signal_filter_stability} follows from \( \mathcal K_\nu(\mathbf I) = L_\nu^*+L_\nu \in\mathcal O. \) It does not require full observability: the estimated filter may reproduce the observed signals asymptotically without reconstructing the entire conditional state. In structured models such as degenerate QND systems, the observable information may be restricted to the sector populations, which motivates the reduced filters considered later.

\subsection{Fidelity-based filter stability}
\label{subsec:fidelity-observer-stability}

A complementary measure of agreement between the true and estimated filters
is the quantum fidelity
$$
F(\rho,\sigma):=\left(\Tr\sqrt{\rho^{1/2}\sigma\rho^{1/2}}\right)^2,
$$
which satisfies $F(\rho,\sigma)\in[0,1]$ and $F(\rho,\sigma)=1$ if and only if $\rho=\sigma$.
The following result provides a general non-divergence property of quantum
filters.
\begin{theorem}[{\cite[Theorem 5]{amini2014stability}}]
\label{thm:fidelity-submartingale}
Assume that the true and estimated filters use the same system model and
the same prescribed control input, and are driven by the same measurement
record. Then,
\(
\bigl(F(\rho_t,\hat\rho_t)\bigr)_{t\geq0}
\)
is a submartingale with respect to the observation filtration: $\mathbb E
\left[
F(\rho_t,\hat\rho_t)
\,\middle|\,
\mathcal F_s^Y
\right]
\geq
F(\rho_s,\hat\rho_s)$ for all $0\leq s\leq t$.
\end{theorem}
Thus, the expected fidelity is nondecreasing. Since the fidelity is bounded,
there exists an
\(\mathcal F_\infty^Y\)-measurable random variable
\(F_\infty\in[0,1]\) such that $\lim_{t\to\infty}F(\rho_t,\hat\rho_t)=F_\infty$
 almost surely and in \(L^1\). However, the submartingale property alone does not imply that \(F_\infty=1\). A stronger conclusion follows when the monitored evolution is asymptotically purifying. In the purely diffusive, fully observed setting considered here, the general purification condition reduces to the following assumption.
\begin{assumption}[Purification condition]
\label{ass:purification} 
Every nonzero orthogonal projector \(\Pi\) satisfying $\Pi(L_\nu+L_\nu^*)\Pi=\lambda_\nu\Pi$ with $\lambda_\nu\in\mathbb R$ has rank one.
\end{assumption}
The Fuchs--van de Graaf inequalities~\cite[Chapter 9.2.3]{nielsen2010quantum} imply $1-\sqrt{F(\rho,\sigma)}\leq\frac{1}{2}\|\rho-\sigma\|_1\leq\sqrt{1-F(\rho,\sigma)}.$
Hence, convergence of the fidelity to one is equivalent to convergence in
trace norm.
\begin{proposition}[{\cite[Proposition 3.1]{amini2021asymptotic}}]
\label{thm:filter-convergence-purification}
Assume that the dynamics~\eqref{eq:multi-homodyne-sme} is fully observed, i.e., $g=0$ and $\eta_\nu=1$ for all $\nu\in[r]$, and that the true and estimated filters use the same model and measurement record. Moreover, suppose that \(\operatorname{supp}(\rho_0)\subseteq\operatorname{supp}(\hat\rho_0)\) and the purification condition in
Assumption~\ref{ass:purification} holds.
Then, $\lim_{t\to\infty}F(\rho_t,\hat\rho_t)=1$ almost surely and in \(L^1\). Consequently, $\lim_{t\to\infty}\|\rho_t-\hat\rho_t\|_1=0$ almost surely and in \(L^1\).
\end{proposition}
Without the purification condition, the limiting fidelity may be strictly smaller than
one. Under additional identifiability and spectral assumptions, the
Ces\`aro means of the true and estimated filters nevertheless have the same
asymptotic limit; see~\cite{amini2021asymptotic,kummerer2004pathwise}.
In a degenerate QND model, the purification condition generally fails. Indeed, if \(\dim\mathcal H_j>1\), then the QND structure gives \( \Pi_j(L_\nu+L_\nu^*)\Pi_j = 2\Re(l_{\nu,j})\Pi_j \) for every observed channel, while \(\Pi_j\) has rank greater than one. Thus, the measurement record can identify the occupied sector without resolving the state within that sector. Then, full-state reconstruction is neither guaranteed nor required, which motivates the reduced QND filters introduced in the next subsection.

\subsection{Reduced QND filter stability under gain mismatch}
\label{subsec:qnd-diffusive-robust-observer}

Throughout this subsection, assume that the true system satisfies the QND structure in Assumption~\ref{ass:QND} with respect to the decomposition \(\mathcal H = \bigoplus_{j=0}^{d}\mathcal H_j. \) The observer is assumed to preserve the same QND structure: its Hamiltonian and unobserved Lindblad operators are block diagonal with respect to this decomposition, and the sector coefficients \(l_{\nu,j}\) of the observed operators are known. However, the effective measurement gains \(\hat\iota_\nu>0\) may differ from their true values \(\iota_\nu\). Under  the QND structure, the block-diagonal Hamiltonian and unobserved Lindblad operators do not enter the sector-population dynamics. Hence, sector identification can be performed by a reduced filter driven by the actual measurement record, without reconstructing the complete conditional state.

The reduced filter satisfies
\begin{equation}
d\hat q_{j,t}= 2\hat q_{j,t}\sum_{\nu=1}^{r}\hat\iota_\nu\Psi_\nu^j(\hat q_t)\left(dY_{\nu,t}-2\hat\iota_\nu a_\nu(\hat q_t)dt\right), \quad j=0,\ldots,d,
\label{eq:estimated_reduced_qnd_filter}
\end{equation}
where $a_\nu$ and $\Psi^j_\nu$ are defined in~\eqref{eq:a_nuPsi}. By using~\eqref{eq:qnd_reduced_output}, this equation can equivalently be written as
\begin{equation}
\begin{aligned}
d\hat q_{j,t} =  2\hat q_{j,t} \sum_{\nu=1}^{r} \hat\iota_\nu \Psi_\nu^j(\hat q_t) \big[dW_{\nu,t}+2\big(\iota_\nu a_\nu(q_t)-\hat\iota_\nu a_\nu(\hat q_t)\big)dt\big].
\end{aligned}
\label{eq:estimated_reduced_qnd_filter_innovation}
\end{equation}
Since \(\sum_{j} \hat q_j\Psi_\nu^j(\hat q) = 0\),  the normalization \(\sum_{j}\hat q_{j,t}=1\) is preserved. Using the stochastic-exponential arguments~\cite{protter2004stochastic}, for every \(\hat q_0\in\boldsymbol\Delta_{d}^{\circ}\), equation~\eqref{eq:estimated_reduced_qnd_filter} admits a unique global strong solution satisfying \( \hat q_t\in\boldsymbol\Delta_{d}^{\circ}\) for all  $t\geq0$ almost surely.

For \(\nu\in[r]\), set
\[
\chi_\nu :={\iota_\nu}/{\hat\iota_\nu}, \quad \chi:=(\chi_1,\ldots,\chi_r).
\]
For \(i\neq j\), define
\begin{align*}
&\Phi_{i\mid j}^{\nu}(\chi_\nu)
:=
\big(
\Re(l_{\nu,i})
-
\Re(l_{\nu,j})
\big)
\big(
\Re(l_{\nu,i})
+
(1-2\chi_\nu)\Re(l_{\nu,j})
\big)\\
&\Lambda_{i\mid j}(\chi)
:=
2
\sum_{\nu=1}^{r}
\hat\iota_\nu^2
\Phi_{i\mid j}^{\nu}(\chi_\nu),
\end{align*}
We impose the following condition to ensure that sector distinguishability is preserved under gain mismatch.
\begin{assumption}[Parameter-mismatch distinguishability]
\label{ass:parameter-mismatch-distinguishability}
For every \(i\neq j\),
\(
\Lambda_{i\mid j}(\chi)>0.
\)
\end{assumption}
The condition is directional: \(\Lambda_{i\mid j}(\chi)\) quantifies the exponential rejection of observer sector \(i\) when the physical trajectory selects sector \(j\). The following theorem is adapted from~\cite[Theorems 3.1 and 3.3]{liang2024parameter} and reformulated in the notation of the present review.
\begin{theorem}[{\cite[Theorems 3.1 and 3.3]{liang2024parameter}}]
\label{thm:qnd_robust_reduced_observer}
Assume that Assumptions~\ref{ass:QND}, \ref{ass:measurement-distinguishability}, and \ref{ass:parameter-mismatch-distinguishability} hold, and let
\(
\hat q_0\in\boldsymbol\Delta_{d}^{\circ}.
\)
Denote by \(\mathcal R\in\{0,\ldots,d\}\)  the random sector selected by the
true trajectory, as defined in Theorem~\ref{Thm:QSR}. Then, for every
\(j\in\{0,\ldots,d\}\) and every \(i\neq j\),
\begin{equation}
\lim_{t\to\infty}
\frac{1}{t}
\log
\frac{\hat q_{i,t}}{\hat q_{j,t}}
=
-\Lambda_{i\mid j}(\chi)
<0,
\qquad
\mathbb P\text{-a.s. on }\{\mathcal R=j\}.
\label{eq:qnd-observer-pairwise-exponent}
\end{equation}
In particular, for every \(j=0,\ldots,d\),
\begin{equation}
\mathbb P\left(\hat q_{j,t}\to1\right)=\mathbb P(\mathcal R=j)=q_{j,0}=\Tr(\Pi_j\rho_0).
\label{eq:qnd-observer-selection-probability}
\end{equation}
\end{theorem}

\begin{proof}
Fix \(j\) such that \(q_{j,0}>0\). Since \(q_{j,t}\) is a bounded nonnegative martingale, the probability measures \(\mathbb Q^j\) defined by $\left.{d\mathbb Q^j}/{d\mathbb P}\right|_{\mathcal F_t^Y}={q_{j,t}}/{q_{j,0}}$ form a consistent family~\cite{benoist2014large}. By Theorem~\ref{Thm:QSR}, $q_{j,t}$ converges to $\mathds 1_{\{\mathcal R=j\}}$ in \(L^1\). Hence, \(\mathbb Q^j\) extends to \(\mathcal F_\infty^Y\) with ${d\mathbb Q^j}/{d\mathbb P}={\mathds 1_{\{\mathcal R=j\}}}/{q_{j,0}}$, which implies $\mathbb Q^j=\mathbb P(\,\cdot\mid\mathcal R=j).$ By Girsanov's theorem~\cite{ikeda2014stochastic}, the processes
\begin{equation*}
\overline W_{\nu,t}^{\,j}
:=
W_{\nu,t}
-
2\iota_\nu
\int_0^t
\Psi_\nu^j(q_s)\,ds,
\qquad
\nu\in[r],
\end{equation*}
are independent standard Wiener processes under \(\mathbb Q^j\).
By applying It\^o's formula, for \(i\neq j\), we obtain
\begin{equation*}
\begin{aligned}
\log \frac{\hat q_{i,t}}{\hat q_{j,t}} =\log\frac{\hat q_{i,0}}{\hat q_{j,0}} -\Lambda_{i\mid j}(\chi)t +2\sum_{\nu=1}^{r} \hat\iota_\nu [\Re(l_{\nu,i})-\Re(l_{\nu,j})]
\overline W_{\nu,t}^{\,j}.
\end{aligned}
\end{equation*}
Since ${\overline W_{\nu,t}^{\,j}}/{t}$ converges to zero $\mathbb Q^j$ almost surely, we obtain
\[
 \lim_{t\to\infty}\frac{1}{t}\log\frac{\hat q_{i,t}}{\hat q_{j,t}}=-\Lambda_{i\mid j}(\chi),\quad \mathbb Q^j\text{-a.s}.
\]
Under the parameter-mismatch distinguishability condition in Assumption~\ref{ass:parameter-mismatch-distinguishability}, every ratio \(\hat q_{i,t}/\hat q_{j,t}\) for \(i\neq j\) converges exponentially to zero on \(\{\mathcal R=j\}\). Using
\(
\hat q_{j,t}=(1+\sum_{i\neq j}{\hat q_{i,t}}/{\hat q_{j,t}})^{-1},
\)
we conclude that $\hat q_{j,t}$ converges to one $\mathbb{P}$-almost surely on $\{\mathcal R=j\}.$
Then, Theorem~\ref{Thm:QSR} yields
\(
\mathbb P(\mathcal R=j)=q_{j,0}=\Tr(\Pi_j\rho_0),
\)
which proves~\eqref{eq:qnd-observer-selection-probability}.
\end{proof}
In the calibrated case,  $\chi_\nu=1$ for all $\nu\in[r]$ and \(\Phi_{i\mid j}^{\nu}(1)=[\Re(l_{\nu,i})-\Re(l_{\nu,j})]^2.\)
Therefore, we deduce $\Lambda_{i\mid j}(1)=2\sum_{\nu}\iota_\nu^2[\Re(l_{\nu,i})-\Re(l_{\nu,j})]^2>0$ by the measurement-distinguishability condition. Exact calibration is thus unnecessary: the reduced filter
remains sector-consistent whenever the pairwise rejection rates in the parameter-mismatch distinguishability condition in
Assumption~\ref{ass:parameter-mismatch-distinguishability} remain positive.
The theorem establishes convergence of the estimated sector populations,
not convergence of the full density operators. This is sufficient for QND
feedback laws depending only on the sector probabilities. The reduced filter
is represented by \(d+1\) populations, subject to one normalization
constraint, rather than by a full density matrix.

Filter stability under a prescribed input does not by itself yield a separation principle for nonlinear quantum feedback~\cite{bouten2008separation,gough2017non}. Once the estimated state is used in the control law, the physical and observer states form a coupled closed-loop stochastic system, which can be analyzed directly.

\section{Measurement-based feedback stabilization}
\label{sec:feedback-stabilization}

Open-loop QND measurements drive the conditional state toward the union of the invariant sectors associated with the measurement decomposition. The limiting sector is selected randomly, with probabilities determined by the initial sector populations; see Theorem~\ref{Thm:QSR}. Measurement-based feedback modifies this asymptotic behavior by combining the information and local contraction induced by the measurement with a control Hamiltonian that destabilizes undesired invariant sectors and selects a prescribed target. We first consider the ideal state-feedback setting, in which the initial state and all model parameters are known. The conditional state \(\rho_t\) can be reconstructed exactly from the observation record and used in a feedback law \(u_t=u(\rho_t).\) Although \(\rho_t\) is not directly measured, it is adapted to the observation filtration and serves as the information state available to the controller. This architecture is commonly referred to as Bayesian feedback.

The stabilization objective is to render a prescribed pure state
\(\boldsymbol\rho_{\bar n}\), or more generally the target set
\(\mathcal I(\mathcal H_0)\), globally attractive and stable for the
closed-loop SME. We first review global Lyapunov constructions for
low-dimensional systems, then switching strategies combining global escape
with local stabilization, and finally almost-sure exponential stabilization
of QND targets.

\subsection{Lyapunov-based feedback design for low-dimensional systems}
\label{subsec:early-lyapunov-feedback}

Early measurement-based feedback studies showed how continuous monitoring
can be combined with Hamiltonian control to exploit measurement back-action
and steer the conditional state toward a prescribed target. In particular,
continuous measurements were combined with state-based feedback
to deterministically prepare Dicke states in~\cite{stockton2004deterministic}.
A stochastic control-theoretic formulation of feedback-controlled quantum
state reduction was subsequently developed by van Handel, Stockton, and
Mabuchi~\cite{van2005feedback}. Their approach combines measurement-induced
state reduction with Hamiltonian destabilization of undesired equilibria,
while a stochastic Lyapunov argument establishes convergence to a prescribed
measurement eigenstate.

Consider a qubit under continuous measurement of \(L=\sigma_z\), with
feedback acting through the noncommuting control Hamiltonian \(H_1=\sigma_y\), up to
normalization. In the absence of feedback, the eigenstates $\boldsymbol\rho_g=\operatorname{diag}(1,0)$ and $\boldsymbol\rho_e=\operatorname{diag}(0,1)$ are equilibria of the QND open-loop dynamics, and the trajectory converges to either state with probabilities determined by the corresponding initial populations; see Theorem~\ref{Thm:QSR}.

Suppose that \(\boldsymbol\rho_e\) is the prescribed target. A natural
feedback design satisfies $u(\boldsymbol\rho_e)=0$ and $u(\boldsymbol\rho_g)\neq0.$
The first condition preserves the target equilibrium. Since $[\sigma_y,\boldsymbol\rho_g]\neq0,$ the second activates the Hamiltonian at \(\boldsymbol\rho_g\), which is  no longer a closed-loop
equilibrium. However, this equilibrium calculation alone does not imply global convergence, since
other invariant or recurrent sets may remain. To address this issue, we construct a function \(V\) such that
$V(\rho)\geq0$ and $V(\rho)=0$ if and only if $\rho=\boldsymbol\rho_e$, $\mathscr L_uV(\rho)\leq0,$ where \(\mathscr L_u\) denotes the infinitesimal generator of the closed-loop SME. If the largest invariant subset of
\(
\{\rho\in\mathcal S(\mathcal H):\mathscr L_uV(\rho)=0\}
\)
is \(\{\boldsymbol\rho_e\}\),  the stochastic LaSalle invariance principle
can be used to establish convergence to the prescribed target
~\cite{kushner1968concept,mao1999stochastic}.

For a qubit, the Bloch-ball representation converts the stability conditions
into polynomial inequalities on a low-dimensional semialgebraic set.
Sum-of-squares relaxations and semidefinite programming can then assist in the
construction and verification of polynomial Lyapunov certificates~\cite{van2005feedback}. Related
symmetry-reduced constructions were subsequently used for entangled-state
preparation in two-spin systems~\cite{yamamoto2007feedback}.

The principal limitation of this approach is scalability. The global
Lyapunov certificate is model dependent, and its computational construction
becomes increasingly difficult as the system dimension grows. This motivated
the switching strategies considered next, in which global escape and local
stabilization are established separately.

\subsection{Switching feedback via global escape and local stabilization}
\label{subsec:support-theorem-switching-feedback}

In~\cite{mirrahimi2007stabilizing}, a switching strategy is introduced to avoid the construction of a single strict Lyapunov
function on the entire state space. The method combines a global escape
mechanism, established through the support theorem~\cite{stroock1972support}, with a local stabilizing
feedback analyzed by the stochastic LaSalle invariance principle.

Consider the \(N\)-level angular-momentum system with QND measurement operator
\[
J_z
=
\operatorname{diag}(J,J-1,\ldots,-J),
\qquad
J=(N-1)/2,
\]
and control Hamiltonian \(H_1=J_y\), which  couples adjacent eigenstates of \(J_z\),
\[
J_y=
\left[\begin{smallmatrix}
0&-ic_1 &&&\\
ic_1&0&-ic_2&&\\
&\ddots&\ddots&\ddots&\\
&&ic_{N-2}&0&-ic_{N-1}\\
&&&ic_{N-1}&0
\end{smallmatrix}\right],
\qquad
c_m=\frac12\sqrt{m(N-m)} .
\]
Let
\(
\bar E:=\{\boldsymbol\rho_1,\dots,\boldsymbol\rho_N\}
\)
denote the rank-one eigenprojectors of \(J_z\), and define
\(
q_n(\rho)
:=
\operatorname{Tr}(\boldsymbol\rho_n\rho).
\)
Every element of \(\bar E\) is an equilibrium of the QND open-loop  dynamics,
and the conditional state converges to a random element of \(\bar E\); see
Theorem~\ref{Thm:QSR}. The feedback objective is to select a prescribed
target \(\boldsymbol\rho_{\bar n}\in\bar E\).

\paragraph{Local stabilizing mode.}
Near the target, consider the candidate Lyapunov function
\(
V_{\bar n}(\rho):=1-q_{\bar n}(\rho)
\)
and the feedback law
\begin{equation*}
u_{\bar n}(\rho):=-\Tr\bigl(i[J_y,\rho]\boldsymbol\rho_{\bar n}\bigr).
\end{equation*}
Since
\(
u_{\bar n}(\boldsymbol\rho_{\bar n})=0,
\)
the target remains an equilibrium. Moreover, the infinitesimal generator of
the locally controlled SME satisfies
$
\mathscr L_{u_{\bar n}}V_{\bar n}(\rho)=-u_{\bar n}(\rho)^2\leq0.
$
This estimate gives local stability in probability, although its
zero-generator set may contain points other than the target.
For the invariance analysis, introduce
\(
\widetilde V_{\bar n}(\rho):=1-q_{\bar n}(\rho)^2.
\)
Since
\(
J_z\boldsymbol\rho_{\bar n} = \lambda_{\bar n}\boldsymbol\rho_{\bar n},
\) by a straightforward computation, we obtain
\begin{equation*}
\mathscr L_{u_{\bar n}} \widetilde V_{\bar n}(\rho) =-2q_{\bar n}(\rho)u_{\bar n}(\rho)^2-4\eta\gamma q_{\bar n}(\rho)^2 \big(\lambda_{\bar n}-\Tr(J_z\rho)\big)^2\leq0.
\end{equation*}
For a suitable \(\delta>0\), the largest invariant subset of
\(
\{\rho:q_{\bar n}(\rho)>{\delta}/{2},\mathscr L_{u_{\bar n}}\widetilde V_{\bar n}(\rho)=0\}
\)
is \(\{\boldsymbol\rho_{\bar n}\}\). Then, the stochastic LaSalle invariance
principle yields convergence to the target for trajectories that remain
in the local-control region.

\paragraph{Global escape mode.}
Away from the target, a constant nonzero control, normalized as \(u\equiv1\),
is applied. Its role is to drive the process away from the zero-target
population face
\[
\mathcal Z_{\bar n} := \left\{ \rho\in\mathcal S(\mathcal H): q_{\bar n}(\rho)=0 \right\},
\]
which contains all undesired measurement eigenstates.

The support theorem relates the support of the SME trajectories to the
solutions of the associated deterministic control system obtained from its
Stratonovich form~\cite{ikeda2014stochastic}. For the angular-momentum model, the simple spectrum of
\(J_z\) and the connected nearest-neighbor couplings generated by \(J_y\)
yield accessibility of states with strictly positive target population from
\(\mathcal Z_{\bar n}\).
Then, the support theorem and continuous dependence on the initial condition
give a positive probability of leaving a sufficiently small neighborhood of
\(\mathcal Z_{\bar n}\) within finite time. Compactness of $\{\rho:
q_{\bar n}(\rho)\leq\delta/2\}$ yields constants \(T_\delta>0\) and \(\zeta_\delta>0\) such that
\begin{equation*}
 \inf_{\rho\in\{q_{\bar n}(\rho)\leq\delta/2\}}\mathbb P_\rho\left(\sup_{0\leq t\leq T_\delta}q_{\bar n}(\rho_t)\geq\delta\right)\geq\zeta_\delta.
\end{equation*}
Iteration of the above estimate through the strong Markov property~\cite{ikeda2014stochastic}
gives almost-sure entrance into the local-control region and a uniform bound
on the corresponding expected hitting time.

\paragraph{Hysteresis switching.}
Define $\mathcal R_{\mathrm{loc}}:=\{\rho:q_{\bar n}(\rho)\geq\delta\}$ and $\mathcal R_{\mathrm{esc}}:=\{\rho:q_{\bar n}(\rho)\leq\delta/2\}$, and the buffer region
\(
\mathcal B_\delta:=\left\{\rho:{\delta}/{2}<q_{\bar n}(\rho)<\delta\right\}.
\)
Introduce a mode variable $\kappa_t \in \{\mathrm{loc},\mathrm{esc}\}$. The state-based hysteresis switching law is defined as
\begin{equation}
\label{eq:MVH-switching-feedback}
u_t=
\begin{cases}
u_{\bar n}(\rho_t),& \kappa_t=\mathrm{loc},\\
1,& \kappa_t=\mathrm{esc}.
\end{cases}, \qquad 
\kappa_t=\begin{cases}
\mathrm{loc}, & \rho_t\in\mathcal R_{\mathrm{loc}},\\
\mathrm{esc}, & \rho_t\in\mathcal R_{\mathrm{esc}},\\
\kappa_{t^-}, & \rho_t\in\mathcal B_\delta.
\end{cases}
\end{equation}
If the initial state lies in \(\mathcal B_\delta\), either mode may be chosen
initially. The buffer region prevents arbitrarily rapid switching and permits
the closed-loop process to be constructed by concatenating the two smooth
stochastic dynamics.

The escape estimate ensures repeated entrance into the local-control region.
Local stability provides a strictly positive probability that, after such an
entrance, the trajectory remains in the larger region
\(\{q_{\bar n}>\delta/2\}\) and converges to the target. Repeated application
of the strong Markov property shows that one of these stabilization attempts
succeeds almost surely.

\begin{theorem}[{\cite[Theorem 4.2]{mirrahimi2007stabilizing}}]
\label{thm:MVH-switching-GAS}
Consider the \(N\)-level angular-momentum SME~\eqref{eq:single-homodyne-sme} with QND measurement operator
\(J_z\), control Hamiltonian \(J_y\). For every prescribed eigenstate
\(\boldsymbol\rho_{\bar n}\in\bar E\), there exists \(\delta>0\) such that
the hysteresis feedback~\eqref{eq:MVH-switching-feedback} renders
\(\boldsymbol\rho_{\bar n}\) GAS almost surely. In particular, for every
\(\rho_0\in\mathcal S(\mathcal H)\),
\begin{equation*}
\mathbb P_{\rho_0}\left(\lim_{t\to\infty}\|\rho_t-\boldsymbol\rho_{\bar n}\|_1=0\right)=1,\quad 
\lim_{t\to\infty}\mathbb E_{\rho_0}\left[\|\rho_t-\boldsymbol\rho_{\bar n}\|_1\right]=0.
\end{equation*}
\end{theorem}
In contrast with the global Lyapunov construction of
Section~\ref{subsec:early-lyapunov-feedback}, this method applies to
arbitrary finite-dimensional angular-momentum systems without requiring a
single Lyapunov function valid on the entire state space.

In~\cite{ticozzi2012stabilization}, this strategy is extended to stochastic quantum systems with unobserved
Markovian channels, non-Hermitian measurement operators, and target
subspaces. Their analysis first characterizes whether the dissipative
dynamics, possibly supplemented by a time-independent Hamiltonian, can render
the target subspace attractive. When such an open-loop design is unavailable,
a filtering-based switching controller combines a global escape mode with a
local stabilizing mode.
This extension formulates the switching method in terms of invariant
subspaces rather than a particular angular-momentum basis and clarifies the
respective roles of dissipation, Hamiltonian control, and continuous
measurement. The resulting conclusions are qualitative; an explicit convergence exponent is not obtained. The next subsection addresses this issue for QND targets through a global recurrence estimate and a local exponential Lyapunov argument.

\subsection{A unified recurrence--contraction framework for QND state feedback}
\label{subsec:continuous-feedback-GES-QND}

The preceding switching strategy separates the global and local components of stabilization. A similar recurrence--contraction structure applies to a broader class of QND systems: Hamiltonian feedback removes non-target invariant obstructions and provides global recurrence, whereas the measurement dynamics generate local exponential contraction. This subsection reformulates the corresponding results of \cite{liang2019exponential,liang2021GHZ,liang2024model} under a common QND notation and a transparent set of sufficient conditions.

Let $\mathcal H=\mathcal H_0\oplus\mathcal H_R$ with $\mathcal H_R
:=\bigoplus_{j=1}^{d}\mathcal H_j$ with orthogonal projectors \(\Pi_j\), and denote $q_j(\rho):=\Tr(\Pi_j\rho)$ and $s(\rho):=1-q_0(\rho)$. The target set is \(\mathcal I(\mathcal H_0)\).   We introduce  \(d_0(\rho):=\|\rho-\Pi_0\rho\Pi_0\|_1\). Since the state space is
finite dimensional and \(\rho\geq0\), there exists a constant \(c_0>0\)
such that $d_0(\rho)^2\leq c_0 (1-q_0(\rho))$ for all $\rho\in\mathcal S(\mathcal H)$.
In particular, $d_0(\rho)=0$ if and only if $\rho\in\mathcal I(\mathcal H_0)$.
For \(\varepsilon>0\), define
\[
B_\varepsilon(\mathcal H_0):=\{\rho\in\mathcal S(\mathcal H):d_0(\rho)<\varepsilon\},\quad\tau_\varepsilon:=\inf\{t\geq0:\rho_t\in B_\varepsilon(\mathcal H_0)\}.
\]
The proof combines two properties. First, every target neighborhood is reached
almost surely from every initial state. Second, once the trajectory is close
to the target, the QND measurement produces exponential contraction, while
the feedback Hamiltonian is sufficiently weak not to alter the local
exponent. The strong Markov property links these global and local estimates.

\paragraph{Target-compatible QND structure.}

We assume throughout that \(\gamma_\nu>0\) and $\eta_{\nu}\in(0,1]$ for all \(\nu\in[r]\), and assume that the QND structure in Assumption~\ref{ass:QND} holds. This ensures the invariance of \(\mathcal I(\mathcal H_0)\) without control input by Lemma~\ref{lem:open_loop_invariance}.

For \(j\in[d]:=\{1,\ldots,d\}\), define
\[
\delta_{\nu,j}:=\Re(l_{\nu,0})-\Re(l_{\nu,j}).
\]
We impose the following one-sided measurement-separation condition:
\begin{assumption}[One-sided separation]
\label{ass:one-sided-separation}
For every \(j\in[d]\), there exists \(\nu\in[r]\) such that
\(\delta_{\nu,j}\neq0\). Moreover, for each \(\nu\in[r]\), either $\delta_{\nu,j}\geq0$ for all $j\in[d]$ or $\delta_{\nu,j}\leq0$ for all $j\in[d]$.
\end{assumption}
The first requirement separates every non-target sector from the target. The second requires the target measurement value to be extremal in each observed channel and prevents cancellation among the contributions of different non-target populations. In what follows, we restrict attention to targets satisfying the one-sided separation condition. This class includes the extremal angular-momentum eigenstates \(\boldsymbol\rho_1\) and \(\boldsymbol\rho_N\), for which the ordering of the eigenvalues of \(J_z\) implies the required one-sided sign condition. For non-extremal targets, the coefficients \(\delta_{\nu,j}\) may have different signs, and convex combinations of the non-target populations may cancel the measurement direction associated with the target population. Such targets can be treated under alternative nondegeneracy or transversality conditions; see \cite{liang2019exponential,liang2021GHZ,liang2024model}. We adopt Assumption~\ref{ass:one-sided-separation} here since it yields a transparent coercivity condition for the tutorial development below.

Let 
\[
 A_\nu(\pi):=\sum_{j=1}^{d}\delta_{\nu,j}\pi_j,
\]
and define the measurement-coercivity constant
\begin{equation*}
\Lambda_0:= \inf_{\pi\in\boldsymbol\Delta_{d-1}}\sum_{\nu=1}^{r}\iota_\nu^2A_\nu(\pi)^2,  \text{ with }
\boldsymbol\Delta_{d-1}=\{\pi\in[0,1]^d:\textstyle\sum_{j=1}^{d}\pi_j=1\}.
\end{equation*}
The one-sided separation condition in Assumption~\ref{ass:one-sided-separation} implies \(\Lambda_0>0\).

\paragraph{Escape from the zero-target-population face.}
To establish recurrence to the target, one must first exclude trajectories confined to states with zero target population.
Consider the deterministic support system associated with the Stratonovich form of the closed-loop SME~\eqref{eq:multi-homodyne-sme}:
\begin{equation}
\dot\rho_t^v
=
\widetilde{\mathcal L}_{u}(\rho_t^v)
+
\sum_{\nu=1}^{r}
\iota_\nu
\mathcal G_{L_\nu}(\rho_t^v)v_\nu(t),
\label{eq:QND-support-system}
\end{equation}
where \(v=(v_1,\ldots,v_r)\) is locally bounded and 
$$
\widetilde{\mathcal L}_{u}(\rho):=\mathcal L_{u(\rho)}(\rho)-\frac12\sum_{\nu=1}^{r}\iota_\nu^2\nabla\mathcal G_{L_\nu}(\rho)\bigl[\mathcal G_{L_\nu}(\rho)\bigr].
$$
Here, \(\nabla\mathcal G_{L_\nu}(\rho)[X]\) denotes the Fr\'echet derivative in
the direction \(X\).

Define
\[
\mathcal Z_0 := \{\rho\in\mathcal S(\mathcal H):q_0(\rho)=0\}.
\]
Since the measurement operators are diagonal with respect to the measurement decomposition, for every \(\rho\in\mathcal Z_0\) and \(\nu\in[r]\), \( \Tr ( \Pi_0\mathcal G_{L_\nu}(\rho) ) =\Tr( \Pi_0 \nabla\mathcal G_{L_\nu}(\rho)[ \mathcal G_{L_\nu}(\rho) ] )= 0\). Hence, neither the measurement diffusion nor its Stratonovich correction can generate positive target population from \(\mathcal Z_0\). Thus, such obstruction must be removed by the control Hamiltonian $H_1$.

We impose the following conditions on the feedback law.
\begin{assumption}[Feedback structure]
\label{ass:feedback-structure}
\(u\in C^1(\mathcal S(\mathcal H),\mathbb R)\) and $u(\rho)=0$ for all $\rho\in\mathcal I(\mathcal H_0)$ whereas $u(\rho)\neq0$ for all $\rho\in\mathcal Z_0$. Moreover, there exist \(c_u>0\), \(\alpha>1/2\), and a neighborhood
\(\mathcal U_0\) of \(\mathcal I(\mathcal H_0)\) such that $|u(\rho)|\leq c_us(\rho)^\alpha$ for all $\rho\in\mathcal U_0.$
\end{assumption}
The condition \(u=0\) on \(\mathcal I(\mathcal H_0)\) preserves target invariance, whereas \(u\neq0\) on \(\mathcal Z_0\) activates the Hamiltonian control on the undesired zero-feedback configurations. The local growth bound ensures that the Hamiltonian contribution is asymptotically negligible relative to the measurement-induced contraction near the target. These conditions alone do not guarantee escape from \(\mathcal Z_0\). Therefore, we impose the following accessibility condition.
\begin{assumption}[Zero-feedback escape]
\label{ass:zero-feedback-escape}
For every \(\rho_0\in\mathcal Z_0\), there exist a locally bounded support
control \(v\) and \(T>0\) such that the corresponding solution of
\eqref{eq:QND-support-system} satisfies \(q_0(\rho_T^v)>0.\)
\end{assumption}
The zero-feedback escape condition provides the accessibility property required in the recurrence argument. Under the QND structure and the feedback-structure condition, it can be verified through suitable non-decoupling and reachability properties of the control Hamiltonian \(H_1\), formulated below. For an operator \(X\), let \(\mathscr E_\lambda(X)\) denote its eigenspace associated with the
eigenvalue \(\lambda\).
\begin{assumption}[Hamiltonian non-decoupling]
\label{ass:hamiltonian-nondecupling}
For every \(j\in[d]\) and every eigenvalue \(\lambda\) of
\(\Pi_jH_1\Pi_j\),
\[
\textstyle \big(\bigcap_{i\neq j}\ker(\Pi_iH_1\Pi_j)\big )\cap\mathscr E_\lambda(\Pi_jH_1\Pi_j)=\{0\}.
\]
\end{assumption}
\begin{assumption}[Hamiltonian reachability]
\label{ass:hamiltonian-reachability}
There exists \(m\in\mathbb N\) such that
\[
\operatorname{rank}
\bigl[
H_{1,Q},
H_{1,R}H_{1,Q},
\ldots,
H_{1,R}^{m}H_{1,Q}
\bigr]
\geq
\dim\mathcal H_R-1.
\]
\end{assumption}
The Hamiltonian non-decoupling condition excludes eigenvectors of the internal Hamiltonian \(\Pi_jH_1\Pi_j\) that are completely decoupled from all other measurement sectors. In particular, $[H_1,\rho]\neq0$ for all $\rho\in\mathcal I(\mathcal H_j)$ with $j\in[d]$.
The Hamiltonian reachability condition is a Kalman-type reachability condition for the pair \((H_{1,R},H_{1,Q})\). It implies that the subspace of \(\mathcal H_R\) orthogonal to the iterated coupling directions \( H_{1,Q}, H_{1,R}H_{1,Q}, \dots,H_{1,R}^{m}H_{1,Q} \) has dimension at most one. Any density operator supported on such a residual subspace is necessarily pure. 
Under the feedback-structure condition and the Hamiltonian reachability condition, a trajectory of the system~\eqref{eq:multi-homodyne-sme} under the QND structure with $\eta_\nu<1$ initialized at a pure state becomes mixed instantaneously almost surely~\cite{liang2021GHZ,liang2025exploring}. Hence, the possible one-dimensional residual subspace cannot support a persistent trajectory, which explains why the rank condition is allowed to have a defect of one. 

By applying arguments similar to those in~\cite[Appendix C]{liang2025exploring} and \cite[Section IV.C]{liang2025timevarying}, the QND structure, the feedback-structure condition , and the Hamiltonian non-decoupling and reachability conditions imply that the support system~\eqref{eq:QND-support-system} satisfies the zero-feedback escape condition in Assumption~\ref{ass:zero-feedback-escape}. These structural conditions are satisfied by the standard measurement and control operators for the angular-momentum model with \(\boldsymbol\rho_{\bar n}\) as the prescribed target \cite{liang2019exponential}, and for the multiqubit model with a GHZ state as the prescribed target~\cite{liang2021GHZ}.

\paragraph{Recurrence to the target subspace.}

After the support trajectory leaves \(\mathcal Z_0\), the auxiliary controls
associated with the observed diffusion directions can be chosen to increase
the target population. Let $q=(q_0,\ldots,q_d)\in\boldsymbol\Delta_{d}$ with $q_0<1$,
and define the normalized non-target population vector
\[
\pi_j(q) := {q_j}/{(1-q_0)}, \qquad j\in[d].
\]
When $q_0=1$, $\pi(q)$ is assigned an arbitrary fixed value in $\boldsymbol\Delta^{\circ}_{d-1}$. Then \(\pi(q)\in\boldsymbol\Delta_{d-1}\).  By definition~\eqref{eq:a_nuPsi}, one obtains $\Psi_\nu^0(q)=\sum_{j=1}^{d}\delta_{\nu,j}q_j=(1-q_0)A_\nu(\pi(q))$,
which implies 
\begin{equation}\label{eq:Psi-coercivity-QND}
\sum_{\nu=1}^{r}\iota_\nu^2\Psi_\nu^0(q(\rho))^2\geq\Lambda_0s(\rho)^2.
\end{equation}
This estimate provides a uniform direction whenever the state is away from the target subspace.

Along the support dynamics~\eqref{eq:QND-support-system}, the target
population satisfies
\begin{equation} 
 \frac{d}{dt}q_0(\rho_t^v) = b_0^u(\rho_t^v) + 2q_0(\rho_t^v) \sum_{\nu=1}^{r} \iota_\nu \Psi_\nu^0(q(\rho_t^v))v_\nu(t), 
\label{eq:QND-target-population-support} 
\end{equation}
where
\(
b_0^u(\rho):=\Tr\bigl(\Pi_0\widetilde{\mathcal L}_{u}(\rho)\bigr).
\)
Suppose that a support trajectory reaches a state satisfying \(q_0>0\), and
fix \(\varepsilon>0\). For some $a>0$, define the region
\[
\mathcal K_{\varepsilon,a}:=\{\rho\notin B_\varepsilon(\mathcal H_0):q_0(\rho)\geq a\}.
\]
Since $d_0(\rho)^2 \leq c_0s(\rho)$ for some constant $c_0>0$,  $s(\rho)
\geq {\varepsilon^2}/{c_0}$ for \(\rho\in\mathcal K_{\varepsilon,a}\). By~\eqref{eq:Psi-coercivity-QND}, we deduce
\begin{equation}
q_0(\rho)\sum_{\nu=1}^{r}\iota_\nu^2\Psi_\nu^0(q(\rho))^2\geq{a\Lambda_0\varepsilon^4}/{c_0^2},\quad
\rho\in\mathcal K_{\varepsilon,a}.
\label{eq:qnd-support-uniform-coercivity}
\end{equation}
Since \(b_0^u\) is continuous on the compact state space, it is bounded. Choose
the support controls $v_\nu(t)=K\iota_\nu\Psi_\nu^0(q(\rho_t^v))$ for $\nu\in[r]$ where \(K>0\) is sufficiently large. Substitution into
\eqref{eq:QND-target-population-support} yields
\[
\frac{d}{dt}q_0(\rho_t^v)= b_0^u(\rho_t^v)+2Kq_0(\rho_t^v)\sum_{\nu=1}^{r}\iota_\nu^2\Psi_\nu^0(q(\rho_t^v))^2.
\]
Estimate~\eqref{eq:qnd-support-uniform-coercivity} shows that \(K\) can be
chosen so that
\(
\frac{d}{dt}q_0(\rho_t^v)>0
\)
uniformly on \(\mathcal K_{\varepsilon,a}\). Therefore, the trajectory  enters
\(B_\varepsilon(\mathcal H_0)\) in finite time. This result is adapted from~\cite[Lemma 3.4]{liang2024model}, we restate the argument and the proof in the present notation.
\begin{lemma}[{\cite[Lemma 3.4]{liang2024model}}]
\label{lem:QND-GES-recurrence}
Assume that Assumptions~\ref{ass:QND} and \ref{ass:one-sided-separation}--\ref{ass:zero-feedback-escape}  
hold. Then, for every \(\varepsilon>0\) and for every \(\rho_0\in\mathcal S(\mathcal H)\),
\(
\mathbb P_{\rho_0}\left(\tau_\varepsilon<\infty\right)=1.
\)
\end{lemma}

\begin{proof}
By the support-system construction above,  the zero-feedback escape condition in Assumption~\ref{ass:zero-feedback-escape} and the one-sided separation condition in Assumption~\ref{ass:one-sided-separation} imply that, from every initial state, there exists a trajectory of the support system entering \(B_{\varepsilon/2}(\mathcal H_0)\) in finite time. Then, the support theorem implies that, for every \(\rho\notin B_\varepsilon(\mathcal H_0)\), the stochastic trajectory has a strictly positive probability of entering \(B_\varepsilon(\mathcal H_0)\) within finite time.
Due to the continuous dependence on the initial condition and compactness of \(\mathcal S(\mathcal H)\setminus B_\varepsilon(\mathcal H_0)\),  a finite-cover argument~\cite[Proposition 4.6]{baxendale1990invariant} yields constants \(T_\varepsilon>0\) and \(\zeta_\varepsilon>0\), independent of the initial state, such that
$$
 \inf_{\rho\notin B_\varepsilon(\mathcal H_0)}\mathbb P_\rho\left(\tau_\varepsilon<T_\varepsilon\right)\geq\zeta_\varepsilon.
$$
By applying the Markov property at \(nT_\varepsilon\), we deduce that
\begin{equation*}
\sup_{\rho_0\in\mathcal S(\mathcal H)}\mathbb P_{\rho_0}\left(\tau_\varepsilon>nT_\varepsilon\right)\leq(1-\zeta_\varepsilon)^n,\quad
n\in\mathbb N.
\end{equation*}
Letting \(n\to\infty\) yields \( \tau_\varepsilon<\infty\) almost surely, which proves the result.
\end{proof}

\paragraph{Local exponential Lyapunov estimate.}

For \(x\in(0,1)\), define
\begin{equation*}
V_x(\rho):=s(\rho)^x.
\end{equation*}
The function \(V_x\) vanishes exactly on \(\mathcal I(\mathcal H_0)\). 
The drift of \(s(\rho)\) is given by  
$$\mathfrak F_u(\rho):=i u(\rho)\Tr\bigl(\Pi_0[H_1,\rho]\bigr).$$
Positivity of \(\rho\) implies that there exists \(c_H>0\) such that $|\Tr(\Pi_0[H_1,\rho])|\leq c_H\sqrt{s(\rho)}.$ Then, the feedback-structure condition in Assumption~\ref{ass:feedback-structure} gives ${|\mathfrak F_u(\rho)|}/{s(\rho)}\leq c_Hc_us(\rho)^{\alpha-1/2}$.
A direct generator calculation yields
\begin{equation*}
\frac{\mathscr LV_x(\rho)}{V_x(\rho)}\leq x\frac{\mathfrak F_u(\rho)}{s(\rho)}-
2x(1-x)q_0(\rho)^2\sum_{\nu=1}^{r}\left(\frac{\iota_\nu\Psi_\nu^0(q(\rho))}{s(\rho)}
\right)^2.
\end{equation*}
Combining the above estimate with~\eqref{eq:Psi-coercivity-QND}, we deduce
\begin{equation}
 \limsup_{\rho\to\mathcal I(\mathcal H_0),\,s(\rho)>0}\frac{\mathscr LV_x(\rho)}{V_x(\rho)}\leq-2x(1-x)\Lambda_0<0.
\label{eq:LV-over-V-limsup-QND}
\end{equation}
The diffusion term in the logarithmic Lyapunov estimate satisfies
\begin{equation*}
\sum_{\nu=1}^{r}\left(\frac{\nabla V_x(\rho)\bigl[\iota_\nu\mathcal G_{L_\nu}(\rho)\bigr]}{V_x(\rho)}\right)^2=4x^2q_0(\rho)^2\sum_{\nu=1}^{r}\left(\frac{\iota_\nu\Psi_\nu^0(q(\rho))}{s(\rho)}\right)^2,
\end{equation*}
which implies
\begin{equation}
\liminf_{\rho\to\mathcal I(\mathcal H_0),\, s(\rho)>0}\sum_{\nu=1}^{r}\left(\frac{\nabla V_x(\rho)\bigl[\iota_\nu\mathcal G_{L_\nu}(\rho)\bigr]}{V_x(\rho)}\right)^2\geq
4x^2\Lambda_0.
\label{eq:log-diffusion-lower-bound-QND}
\end{equation}
The estimates~\eqref{eq:LV-over-V-limsup-QND} and
\eqref{eq:log-diffusion-lower-bound-QND} provide the two local ingredients
needed for pathwise exponential stabilization: the former yields local
stability in probability, while the latter strengthens the logarithmic
Lyapunov estimate to an almost-sure exponential rate. Combined with the
global recurrence property of Lemma~\ref{lem:QND-GES-recurrence}, they lead
to the following recurrence--contraction result.
The theorem below provides a tutorial synthesis of the common
recurrence--contraction mechanism underlying the state-feedback
stabilization results developed in
\cite[Theorems~6.3 and~6.4]{liang2019exponential},
\cite[Theorem~12]{liang2021GHZ}, and
\cite[Theorem~3.5]{liang2024model}, under the transparent sufficient conditions adopted here.
\begin{theorem}
\label{thm:QND-GES}
Assume that  Assumptions~\ref{ass:QND} and \ref{ass:one-sided-separation}--\ref{ass:zero-feedback-escape}  
hold. Then, \(\mathcal I(\mathcal H_0)\) is  almost surely  GES for the closed-loop SME~\eqref{eq:multi-homodyne-sme}.
Moreover, for every $\rho_0\in\mathcal S(\mathcal H)\setminus\mathcal I(\mathcal H_0),$
\begin{equation}
\label{eq:qnd-state-feedback-exponent}
\limsup_{t\to\infty}\frac{1}{t}\log d_0(\rho_t)\leq-\Lambda_0,\quad a.s.
\end{equation}
\end{theorem}

\begin{proof}
Fix \(x\in(0,1)\). By~\eqref{eq:LV-over-V-limsup-QND}, there exist a
neighborhood \(\mathcal U\) of \(\mathcal I(\mathcal H_0)\) and \(c_x>0\) such that $\mathscr L V_x(\rho)\leq-c_xV_x(\rho)$ for all $\rho\in\mathcal U\setminus \mathcal I(\mathcal H_0)$. Then, applying a stopped-supermartingale argument~\cite[Theorem 4.2.2]{mao2007stochastic} implies local stability in probability of $\mathcal I(\mathcal H_0)$.

Choose a neighborhood \(\mathcal U_1\) of \(\mathcal I(\mathcal H_0)\) such that
\(\overline{\mathcal U_1}\subset\mathcal U\) and, for some \(p>0\),
\[
\inf_{\rho\in\mathcal U_1}
\mathbb P_{\rho}
\left(
\rho_t\in\mathcal U
\ \text{for all }t\geq0
\right)
\geq p.
\]
Define recursively the entrance times \(\sigma_k\) into
\(\mathcal U_1\) and the subsequent exit times \(\tau_k\) from
\(\mathcal U\). By Lemma~\ref{lem:QND-GES-recurrence},
\(\sigma_1<\infty\) almost surely and, on
\(\{\tau_k<\infty\}\), one has
\(\sigma_{k+1}<\infty\) almost surely. Hence, by the strong Markov
property, we have
\[
\mathbb P(\tau_1<\infty,\ldots,\tau_n<\infty)
\leq
(1-p)^n,
\]
which implies 
\begin{equation*}
\mathbb P\left(\textstyle \bigcup_{k\geq1} E_k \right)=1, \quad E_k:=\{\sigma_k<\infty,\ \tau_k=\infty\}.
\end{equation*}

Fix \(\delta>0\). By
\eqref{eq:LV-over-V-limsup-QND} and
\eqref{eq:log-diffusion-lower-bound-QND}, \(\mathcal U\) may be chosen sufficiently small so that
\begin{equation*}
\begin{aligned}
\frac{\mathscr L V_x(\rho)}{V_x(\rho)}
&-
\frac12
\sum_{\nu=1}^{r}
\left|
\frac{
\nabla V_x(\rho)
[\iota_\nu\mathcal G_{L_\nu}(\rho)]
}{
V_x(\rho)
}
\right|^2
\leq
-2x\Lambda_0+\delta,
\quad
\forall \rho\in\mathcal U\setminus\mathcal I(\mathcal H_0).
\end{aligned}
\end{equation*}
Indeed, the two limiting contributions are respectively bounded by
\(-2x(1-x)\Lambda_0\) and \(-2x^2\Lambda_0\).

By the finite-time non-attainment argument of \cite[Lemma 4.2]{liang2019exponential}, which applies under the present QND multiplicative structure, a trajectory initialized outside $\mathcal I(\mathcal H_0)$ does not reach $\mathcal I(\mathcal H_0)$ at finite time almost surely. Hence, $V_x(\rho_t)>0$ almost surely for every finite $t$.
For each \(k\), applying It\^o's formula to \(\log V_x(\rho_t)\) on \([\sigma_k,t\wedge\tau_k]\) yields 
\[
\log V_x(\rho_{t\wedge\tau_k})
-
\log V_x(\rho_{\sigma_k})
\leq
(-2x\Lambda_0+\delta)
\bigl((t\wedge\tau_k)-\sigma_k\bigr)
+
M_{t\wedge\tau_k}-M_{\sigma_k},
\]
where \(M_t\) is a continuous local martingale. By the strong law for continuous local martingales, we deduce that
\[
\limsup_{t\to\infty}
\frac{1}{t}\log V_x(\rho_t)
\leq
-2x\Lambda_0+\delta,
\quad a.s. \text{ on }E_k.
\]
Since
\(
\mathbb P(\bigcup_{k\geq1}E_k)=1
\)
and \(\delta>0\) is arbitrary, we obtain
\[
\limsup_{t\to\infty}
\frac{1}{t}\log V_x(\rho_t)
\leq
-2x\Lambda_0,
\qquad\text{a.s.}
\]
Moreover, since $d_0(\rho)^2 \leq c_0V_x(\rho)^{1/x}$ for some $c_0>0$, we deduce
\[
\begin{aligned}
\limsup_{t\to\infty}
\frac1t\log d_0(\rho_t)
&\leq
\frac{1}{2x}
\limsup_{t\to\infty}
\frac1t\log V_x(\rho_t)\leq
-\Lambda_0,
\qquad\text{a.s.}
\end{aligned}
\]
Together with local stability in probability, this proves the almost-sure
GES of \(\mathcal I(\mathcal H_0)\) and~\eqref{eq:qnd-state-feedback-exponent}.
\end{proof}

The preceding argument separates the stabilization mechanism into global
recurrence and local exponential contraction. Its logical structure is
summarized as follows:
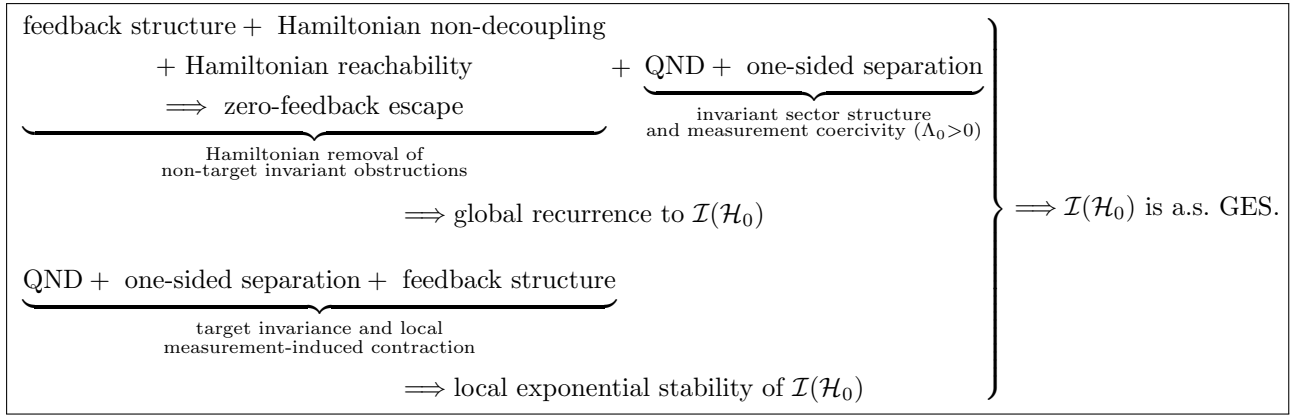
\begin{equation*}
\small
\boxed{
\left.
\begin{aligned}
&
\underbrace{
\begin{gathered}
\text{feedback structure}
+\text{ Hamiltonian non-decoupling}\\
+\text{ Hamiltonian reachability}
\\
\Longrightarrow
\text{ zero-feedback escape}
\end{gathered}
}_{\substack{\text{Hamiltonian removal of}\\
\text{non-target invariant obstructions}}}
+ \
\underbrace{
\begin{gathered}
\text{QND}
+\text{ one-sided separation}
\end{gathered}
}_{\substack{\text{invariant sector structure}\\
\text{and measurement coercivity }(\Lambda_0>0)}}
\\[1mm]
&\hspace{50mm}
\Longrightarrow
\text{global recurrence to }\mathcal I(\mathcal H_0)
\\[3mm]
&
\underbrace{
\begin{gathered}
\text{QND}
+\text{ one-sided separation}
+\text{ feedback structure}
\end{gathered}
}_{\substack{\text{target invariance and local}\\
\text{measurement-induced contraction}}}
\\[1mm]
&\hspace{50mm}
\Longrightarrow
\text{local exponential stability of }\mathcal I(\mathcal H_0)
\end{aligned}
\right\}
\Longrightarrow
\mathcal I(\mathcal H_0)\text{ is a.s.\ GES.}
}
\end{equation*}

Theorem~\ref{thm:QND-GES} highlights the two complementary mechanisms behind the stabilization result. The feedback Hamiltonian provides the global
accessibility needed to eliminate undesired invariant components, whereas the
observed QND channels generate the local contraction and determine the
almost-sure exponential rate through the coercivity constant \(\Lambda_0\).
In particular, the feedback contribution is of higher order near the target
and does not affect the leading exponential rate.
For the \(N\)-level angular-momentum system with \(L=J_z\), \(H_1=J_y\), and prescribed target \(\boldsymbol\rho_1\) or \(\boldsymbol\rho_N\), the nearest-neighbor coupling structure of \(J_y\) provides the required escape mechanism. Under the standard normalization, we have $\Lambda_0=\eta\gamma$. Then, Theorem~\ref{thm:QND-GES} recovers the exponential rate obtained in
\cite{liang2019exponential}. The stabilization of GHZ states follows the same
global-accessibility and local-contraction principle \cite{liang2021GHZ}.

\section{Robust reduced-order observer-based feedback under QND measurements}
\label{subsec:robust-observer-feedback}

Exact state feedback requires knowledge of the physical initial state and all
model parameters. As shown in
Subsection~\ref{subsec:qnd-diffusive-robust-observer}, under QND measurements
the feedback can instead be constructed from the reduced population estimate
\(\hat q_t\), without reconstructing the full conditional state. Here, we consider the dynamics~\eqref{eq:multi-homodyne-sme} under the QND structure where $H_{0,j}\in \mathcal{B}_*(\mathcal{H}_j)$ and $C_{k,j}\in \mathcal{B}(\mathcal{H}_j)$ are unknown, which is the QND specialization of the structure-preserving robustness regime described by the robust-invariance condition in
Assumption~\ref{ass:robust-invariance}.
The resulting closed loop is a coupled physical--observer process on
\(
\mathcal S(\mathcal H)\times\boldsymbol\Delta_{d},
\)
driven by the common measurement record. Therefore, the observer evolves with
only \(d+1\) population coordinates, rather than the order-\(N^2\) variables required by a full density-matrix filter, where \(d+1\leq N\).
The physical and observer models share the measurement coefficients
\(l_{\nu,j}\), while their effective measurement gains may differ. Denote the
physical and observer gains by \(\iota_\nu>0\) and
\(\hat\iota_\nu>0\), respectively, and let $\chi=(\chi_1,\ldots,\chi_r)$ with $\chi_\nu={\iota_\nu}/{\hat\iota_\nu}$.
This section builds on the reduced-filter and robust-feedback analysis
developed in~\cite{cardona2018exponential,liang2024model,liang2025reduced,liang2025exploring}.  

\subsection{Reduced observer architecture and feedback information structure}

Starting from the reduced QND filter introduced in Subsection~\ref{subsec:qnd-diffusive-robust-observer}, we incorporate the effect of the feedback through the population-mixing term \(u(\hat q)\Gamma\hat q\). Since the control Hamiltonian under Assumptions~\ref{ass:hamiltonian-nondecupling} and~\ref{ass:hamiltonian-reachability} generally generates inter-sector coherences, the sector populations do not satisfy a closed equation, and \(u(\hat q)\Gamma\hat q\) should be viewed as an auxiliary reduced model rather than the exact population projection of the controlled SME.

The role of \(\Gamma\) is primarily structural. It is chosen so that the
controlled reduced dynamics preserve the probability simplex, while
\(u(e_0)=0\) guarantees invariance of the target vertex \(e_0\).
At every non-target vertex \(e_j\) for \(j\in[d]\), the conditions below ensure
that \(\Gamma e_j\neq0\). Hence, whenever \(u(e_j)>0\), the vertex \(e_j\)
cannot remain an equilibrium of the controlled reduced dynamics.
A convenient choice is a column-conservative Metzler matrix whose
off-diagonal entries reflect the inter-sector couplings induced by \(H_1\),
as in~\cite{cardona2020exponential}. The subsequent analysis uses
\(\Gamma\) only through the conservation, positivity, and nondegeneracy
properties stated below.

\begin{assumption}[Reduced generator]
\label{ass:reduced-generator}
For every \(j=0,\ldots,d\), $\sum_{i=0}^{d}\Gamma_{i,j}=0$, $\Gamma_{i,j}\geq0$ with $i\neq j$ and $\Gamma_{j,j}<0$.
\end{assumption}

\begin{assumption}[Reduced feedback]
\label{ass:reduced-feedback}
$u\in C^1(\boldsymbol\Delta_{d},\mathbb R_+)$, $u(e_0)=0$ and $u(\hat q)>0$ for all $\hat q\neq e_0$. Moreover, there exist \(c_u>0\), \(\beta>1\), and a neighborhood
\(\widehat{\mathcal U}_0\) of \(e_0\) such that $u(\hat q)\leq c_u(1-\hat q_0)^\beta$ for all $\hat q\in\widehat{\mathcal U}_0$.
\end{assumption}
The controlled reduced observer is
\begin{equation}
d\hat q_{n,t}=u(\hat q_t)(\Gamma\hat q_t)_n dt+ 2\hat q_{n,t}\sum_{\nu=1}^{r}\hat\iota_\nu\Psi_\nu^n(\hat q_t)\left[dY_{\nu,t}-2\hat\iota_\nu a_\nu(\hat q_t)dt\right],
\label{eq:reduced-QND-filter}
\end{equation}
for \(n=0,\ldots,d\), and the physical feedback is $u_t=u(\hat q_t).$
Using the output~\eqref{eq:qnd_reduced_output}, this can equivalently be written as
\begin{equation}
 d\hat q_{n,t}=u(\hat q_t)(\Gamma\hat q_t)_n dt+ 2\hat q_{n,t}\sum_{\nu=1}^{r} \hat\iota_\nu\Psi_\nu^n(\hat q_t)\bigl[\mathrm dW_{\nu,t}+2\bigl(\iota_\nu a_\nu(q_t)-\hat\iota_\nu a_\nu(\hat q_t)\bigr)dt\bigr].
\label{eq:reduced-QND-filter-W}
\end{equation}
The reduced-generator condition, together with
\(
\sum_{n=0}^{d}
\hat q_n\Psi_\nu^n(\hat q)=0
\)
and the nonnegativity of \(u\), ensures invariance of the probability simplex. As for the uncontrolled reduced filter, the multiplicative diffusion structure and the inward-pointing drift imply that, for every
\(
(\rho_0,\hat q_0)\in\mathcal S(\mathcal H)\times\boldsymbol\Delta_{d}^{\circ},
\)
the coupled system admits a unique global strong solution satisfying
\(
(\rho_t,\hat q_t) \in \mathcal S(\mathcal H)\times\boldsymbol\Delta_{d}^{\circ},
\)
for all $t\geq0$ almost surely; see~\cite[Section IV.A]{liang2025reduced}. The local estimate in the reduced-feedback condition in Assumption \ref{ass:reduced-feedback} ensures that the reduced control term is negligible at the exponential scale near the target.

\subsection{Global recurrence of the coupled process}

The desired set of the coupled physical--observer process is
\(
\mathcal I(\mathcal H_0)\times\{e_0\}.
\)
Observer-based feedback introduces an additional possible obstruction: the
observer may approach \(e_0\), causing the feedback to vanish, while the
physical state approaches a possible invariant non-target sector
\(\mathcal I(\mathcal H_n)\).
To exclude such false-target configurations, we
impose the following condition.

\begin{assumption}[False-target rejection]
\label{ass:false-target-rejection}
For every \(n\in[d]\), $\Lambda_{0\mid n}(\chi)=2\sum_{\nu=1}^{r}\hat\iota_\nu^2\Phi_{0\mid n}^{\nu}(\chi_\nu)>0.$
\end{assumption}
The rate \(\Lambda_{0\mid n}(\chi)\) is directional: it quantifies the rejection of the observer target sector \(0\) when the physical trajectory
approaches sector \(n\). To make this mechanism explicit, let $\mathsf V_n(\hat q):=-\log\hat q_n.$ If \(\mathcal I(\mathcal H_n)\) is invariant for the physical dynamics with \(u=0\), then the infinitesimal generator satisfies
\begin{equation}
\limsup_{(\rho,\hat q)\to\mathcal I(\mathcal H_n)\times\{e_0\},\,\hat q_n>0}\mathscr L\mathsf V_n(\rho,\hat q)\leq-\Lambda_{0\mid n}(\chi)<0.
\label{eq:false-target-generator}
\end{equation}
The control-dependent terms vanish in this limit by the reduced-feedback condition in Assumption \ref{ass:reduced-feedback}, whereas
the false-target-rejection condition provides the strictly negative contribution.
A stopped It\^o's formula argument shows that a trajectory initialized
with \(\hat q_0\in\boldsymbol\Delta_{d}^{\circ}\) cannot stay indefinitely in a sufficiently small neighborhood of \(\mathcal I(\mathcal H_n)\times\{e_0\}.\)
Although \(e_0\) is an equilibrium of the reduced observer, the false-target-rejection condition prevents trajectories initialized with
\(\hat q_0\in\boldsymbol\Delta_{d}^{\circ}\) from converging to $\mathcal I(\mathcal H_n)\times\{e_0\}$ with $n\in[d]$; see~\cite[Lemma 4.2]{liang2025reduced} and~\cite[Lemma C.3]{liang2025exploring}.

For \(\varepsilon>0\), define $B_\varepsilon(e_0):=\{\hat q\in\boldsymbol\Delta_d:\|\hat q-e_0\|_1<\varepsilon\}$ and
\[
\tau_\varepsilon^{\mathrm{red}}:=\inf\left\{t\geq0:(\rho_t,\hat q_t)\in B_\varepsilon(\mathcal H_0)\times B_\varepsilon(e_0)\right\}.
\]
The recurrence argument combines the support construction of
Lemma~\ref{lem:QND-GES-recurrence} with the rejection of the false-target
estimate. The latter step is essential since the observer evolves in \(\boldsymbol\Delta_{d}^{\circ}\), and positivity of its
coordinates at finite times does not provide a uniform lower bound away from
the boundary. Choose \(r>0\) sufficiently small so that, for every \(n\in[d]\),
the estimate~\eqref{eq:false-target-generator} holds on
\(
\mathcal U_{n,r}:=B_r(\mathcal H_n)\times B_r(e_0).
\)
Define the compact intermediate region
\[
\mathbf{B}_{\varepsilon,r}:=\left(\mathcal S(\mathcal H)\times\boldsymbol\Delta_{d}\right)\setminus\big[\textstyle B_\varepsilon(\mathcal H_0)\times B_\varepsilon(e_0)\cup\bigcup_{n=1}^{d}\mathcal U_{n,r}\big].
\]
On \(\mathbf{B}_{\varepsilon,r}\), the support construction associated with
the Hamiltonian non-decoupling and reachability conditions in Assumptions~\ref{ass:hamiltonian-nondecupling} and \ref{ass:hamiltonian-reachability}, and the one-sided separation condition in
Assumption~\ref{ass:one-sided-separation} allows the coupled
trajectory to reach the target neighborhood; see~\cite[Lemma~4.8]{liang2021robust}. The support theorem implies that, from every point of 
\(
\mathbf B_{\varepsilon,r}\cap\left(\mathcal S(\mathcal H)\times\boldsymbol\Delta_{d}^{\circ}\right),
\)
the stochastic process has a strictly positive probability of entering the
target neighborhood within finite time.
Then, Feller continuity, compactness of \(\mathbf{B}_{\varepsilon,r}\),  and the localization argument developed in \cite{baxendale1990invariant,liang2021robust} yield the Baxendale-type occupation estimate,
\begin{equation}
\sup_{ (\rho,\hat q)\in \mathbf{B}_{\varepsilon,r} \cap \left( \mathcal S(\mathcal H)\times\boldsymbol\Delta_{d}^{\circ} \right)} \mathbb E_{\rho,\hat q} \left[\int_0^{\tau_\varepsilon^{\mathrm{red}}} \mathds 1_{\mathbf{B}_{\varepsilon,r}}(\rho_t,\hat q_t) dt \right] \leq K_{\varepsilon,r},
\label{eq:reduced-occupation-bound}
\end{equation}
for some \(K_{\varepsilon,r}<\infty\); see~\cite{baxendale1990invariant,liang2021robust}. 
The estimate~\eqref{eq:reduced-occupation-bound} bounds uniformly the expected amount of time that the coupled process can spend in $\mathbf B_{\varepsilon,r}$ before reaching the target neighborhood. However, it does not exclude excursions toward one of the false-target neighborhoods \(\mathcal U_{n,r}\), where \(\mathsf V_n=-\log\hat q_n\) have a strictly negative drift by the false-target-rejection condition in Assumption~\ref{ass:false-target-rejection}. Thus, the process cannot accumulate an arbitrarily large amount of time near a false-target configuration. Combining the occupation estimate with the false-target-rejection argument and a standard localization procedure yields global recurrence of the coupled physical--observer system defined by~\eqref{eq:multi-homodyne-sme} and~\eqref{eq:reduced-QND-filter-W}.
\begin{lemma}[{\cite[Proposition~C.5]{liang2025exploring}}]
\label{lem:reduced-recurrence}
Assume that  $\eta_\nu\in(0,1)$ for all $\nu\in[r]$, and  Assumptions~\ref{ass:QND}, \ref{ass:one-sided-separation} and 
\ref{ass:hamiltonian-nondecupling}--\ref{ass:false-target-rejection}
hold.
Then, for every \(\varepsilon>0\) and every
\(
(\rho_0,\hat q_0)\in\mathcal S(\mathcal H)\times\boldsymbol\Delta_{d}^{\circ},
\)
one has
\(
\mathbb P_{\rho_0,\hat q_0}\left(\tau_\varepsilon^{\mathrm{red}}<\infty\right)=1.
\)
\end{lemma}

\begin{proof}[Proof outline.]
The complete proof is given in~\cite[Proposition~C.5]{liang2025exploring}; see also~\cite{baxendale1990invariant,liang2021robust} for the occupation-time and localization arguments. We summarize the main mechanism, which combines false-target rejection with the occupation estimate~\eqref{eq:reduced-occupation-bound}.

For \(n\in[d]\), set
\(
\mathsf V_n(\hat q):=-\log\hat q_n.
\)
By the false-target-rejection condition in Assumption~\ref{ass:false-target-rejection},  there exists \(c_n>0\) such that
\(
\mathscr L\mathsf V_n(\rho,\hat q)\leq -c_n
\)
on
\(
\mathcal U_{n,r}\cap\left(\mathcal S(\mathcal H)\times\boldsymbol\Delta_{d}^{\circ}\right).
\)
Thus, the coupled process cannot remain near such a false-target configuration for an arbitrarily long time without being driven away from the observer state \(e_0\).

Since there are only finitely many non-target sectors,  a cutoff construction combines the local functions \(\mathsf V_n\) into a nonnegative function  \(\mathsf V\) that coincides with \(\mathsf V_n\) near \(\mathcal I(\mathcal H_n)\times\{e_0\}\). Moreover, there exist constants \(c>0\) and \(C_{\varepsilon,r}<\infty\) such that, before the target neighborhood is reached,
\(
\mathscr L\mathsf V \leq -c + C_{\varepsilon,r}\mathds{1}_{\mathbf B_{\varepsilon,r}}.
\)
To localize the logarithmic singularities, define $\sigma_{t,m}:=t\wedge\tau_\varepsilon^{\mathrm{red}}
\wedge\tau_m$ with
\(
\tau_m:=\inf\{t\geq0:\min_{n\in[d]}\hat q_{n,t}\leq1/m\}.
\)
Applying It\^o's formula to \(\mathsf V\) up to \(\sigma_{t,m}\) and using the occupation estimate~\eqref{eq:reduced-occupation-bound}  yields
\[
c\mathbb E_{\rho_0,\hat q_0}[\sigma_{t,m}]\leq\mathsf V(\rho_0,\hat q_0)+C_{\varepsilon,r}K_{\varepsilon,r}.
\]
Since  \(\boldsymbol\Delta_{d}^{\circ}\) is invariant almost surely,  \(\tau_m\uparrow\infty\) almost surely. Letting first \(t\to\infty\) and then \(m\to\infty\) yields
\(
\mathbb E_{\rho_0,\hat q_0}[\tau_\varepsilon^{\mathrm{red}}]<\infty,
\)
which proves the result.
\end{proof}

\subsection{Local exponential Lyapunov estimate}

Set $\hat s(\hat q):=1-\hat q_0$ and $S(\rho,\hat q):=s(\rho)+\hat s(\hat q)$,
and, for \(x\in(0,1)\), define
\begin{equation*}
V_x(\rho,\hat q):=S(\rho,\hat q)^x\geq 0,
\end{equation*}
which vanishes exactly on
\(
\mathcal I(\mathcal H_0)\times\{e_0\}.
\)
The diffusion coercivity of the physical--observer error is measured by
\begin{equation}
\label{eq:reduced-diffusion-constant}
\Lambda_{\mathrm{red}}(\chi):=\inf_{a\in[0,1],\,\pi,\tilde\pi\in\boldsymbol\Delta_{d-1}} \sum_{\nu=1}^{r}\hat\iota_\nu^2\left[\chi_\nu aA_\nu(\pi)+(1-a)A_\nu(\tilde\pi)\right]^2.
\end{equation}
Under the one-sided separation condition in Assumption~\ref{ass:one-sided-separation}, the two terms inside the brackets have the same sign for each channel. Hence, cancellation cannot occur, and compactness gives $\Lambda_{\mathrm{red}}(\chi)>0$ for all $\chi\in(0,\infty)^r$.

The gain mismatch introduces an additional first-order drift. Unlike $\Lambda_{\mathrm{red}}(\chi)$, which depends only on the relative
measurement separations $\delta_{\nu,j}$, the corresponding drift bound
also depends on the reference output level $\Re(l_{\nu,0})$. Therefore, we 
keep fixed throughout this subsection the measurement convention used in
\eqref{eq:qnd_reduced_output} and define
\begin{equation*}
 B_{\mathrm{red}}(\chi):=4\sum_{\nu=1}^{r}\hat\iota_\nu^2|\chi_\nu-1||\Re(l_{\nu,0})|\bar\ell_\nu,\quad
\bar\ell_\nu:=\max_{j\in[d]}|\delta_{\nu,j}|.
\end{equation*}
We impose the following dominance condition.
\begin{assumption}[Contraction dominance]
\label{ass:contraction-dominance}
\(
B_{\mathrm{red}}(\chi)<2\Lambda_{\mathrm{red}}(\chi).
\)
\end{assumption}
This condition ensures that the measurement-induced contraction dominates the drift caused by the gain mismatch. In the nominal case, we deduce
\(
B_{\mathrm{red}}(1)=0.
\)
Thus, Assumption~\ref{ass:contraction-dominance} follows directly from \(\Lambda_{\mathrm{red}}(1)>0\). Since both sides depend
continuously on \(\chi\), the condition remains valid for all sufficiently small fixed gain mismatches.

The drift of \(S\) can be written as
\begin{equation*}
\mathfrak b_S(\rho,\hat q)=\mathfrak F_u(\rho,\hat q)-u(\hat q)(\Gamma\hat q)_0+\mathfrak b_{\mathrm{mis}}(\rho,\hat q),
\end{equation*}
where
\(
\mathfrak F_u(\rho,\hat q):=iu(\hat q)\Tr\bigl(\Pi_0[H_1,\rho]\bigr).
\)
Near
\(\mathcal I(\mathcal H_0)\times\{e_0\}\), a direct expansion yields
\[
\mathfrak b_{\mathrm{mis}}(\rho,\hat q)=-4\sum_{\nu=1}^{r}\hat\iota_\nu^2\hat q_0(\chi_\nu-1)\Re(l_{\nu,0})\Psi_\nu^0(\hat q)+O\bigl(S(\rho,\hat q)^2\bigr).
\]
Since
\(
|\Psi_\nu^0(z)|\leq\bar\ell_\nu(1-z_0),
\)
it follows that
\(
\mathfrak b_{\mathrm{mis}}(\rho,\hat q)\leq B_{\mathrm{red}}(\chi)S(\rho,\hat q)
+O\bigl(S(\rho,\hat q)^2\bigr).
\)
Moreover, the reduced-feedback condition in Assumption \ref{ass:reduced-feedback} and the
Hamiltonian cross-term estimate imply
\[
\frac{|\mathfrak F_u(\rho,\hat q)|}{S(\rho,\hat q)}=O\bigl(S(\rho,\hat q)^{\beta-\frac12}\bigr), \qquad \frac{|u(\hat q)(\Gamma\hat q)_0|}{S(\rho,\hat q)}=O\bigl(S(\rho,\hat q)^{\beta-1}\bigr).
\]
Since \(\beta>1\), both terms vanish near the target. Then, we deduce
\begin{equation}
\limsup_{(\rho,\hat q)\to\mathcal I(\mathcal H_0)\times\{e_0\},\, S(\rho,\hat q)>0}\frac{\mathfrak b_S(\rho,\hat q)}{S(\rho,\hat q)}\leq B_{\mathrm{red}}(\chi).
\label{eq:reduced-drift-limsup}
\end{equation}

For each observed channel, define
$
\mathfrak D_\nu(\rho,\hat q):=\iota_\nu q_0(\rho)\Psi_\nu^0(q(\rho))+\hat\iota_\nu\hat q_0\Psi_\nu^0(\hat q).
$
A direct generator calculation yields
\begin{equation*}
\frac{\mathscr LV_x(\rho,\hat q)}{V_x(\rho,\hat q)}=x\frac{\mathfrak b_S(\rho,\hat q)}{S(\rho,\hat q)}-2x(1-x)\frac{\sum_{\nu=1}^{r}\mathfrak D_\nu(\rho,\hat q)^2}{S(\rho,\hat q)^2}.
\end{equation*}
Set
\(
a(\rho,\hat q):={s(\rho)}/{S(\rho,\hat q)}\in[0,1].
\)
By a straightforward computation, we obtain
\begin{equation*}
\frac{\mathfrak D_\nu(\rho,\hat q)}{S(\rho,\hat q)}=\hat\iota_\nu \Bigl[\chi_\nu q_0(\rho)a(\rho,\hat q) A_\nu(\pi(q(\rho)))+\hat q_0
\bigl(1-a(\rho,\hat q)\bigr)A_\nu(\pi(\hat q))\Bigr].
\end{equation*}
Then, by~\eqref{eq:reduced-diffusion-constant}, we have
\begin{equation}
\liminf_{(\rho,\hat q)\to\mathcal I(\mathcal H_0)\times\{e_0\}, \,S(\rho,\hat q)>0}\frac{\sum_{\nu=1}^{r}\mathfrak D_\nu(\rho,\hat q)^2}{S(\rho,\hat q)^2}\geq\Lambda_{\mathrm{red}}(\chi).
\label{eq:reduced-diffusion-coercivity}
\end{equation}
Combining \eqref{eq:reduced-drift-limsup} and
\eqref{eq:reduced-diffusion-coercivity}, we deduce
\begin{equation}
\limsup_{(\rho,\hat q)\to\mathcal I(\mathcal H_0)\times\{e_0\}, \,S(\rho,\hat q)>0}\frac{\mathscr LV_x(\rho,\hat q)}{V_x(\rho,\hat q)}\leq-x\left[2(1-x)\Lambda_{\mathrm{red}}(\chi)-B_{\mathrm{red}}(\chi)\right].
\label{eq:reduced-LV-limsup}
\end{equation}
Under the contraction-dominance condition in Assumption~\ref{ass:contraction-dominance}, we can choose $0<x<1-\frac{B_{\mathrm{red}}(\chi)}{2\Lambda_{\mathrm{red}}(\chi)}$ for which the right-hand side of \eqref{eq:reduced-LV-limsup} is strictly negative.

Let \(\Sigma_\nu\) denote the joint diffusion vector field associated with
\(W_\nu\). Then, we have
\[
\frac{\nabla V_x(\rho,\hat q)\bigl[\Sigma_\nu(\rho,\hat q)\bigr]}{V_x(\rho,\hat q)}=-2x\frac{\mathfrak D_\nu(\rho,\hat q)}{S(\rho,\hat q)},
\]
which implies
\begin{equation}
\liminf_{(\rho,\hat q)\to\mathcal I(\mathcal H_0)\times\{e_0\}, \, S(\rho,\hat q)>0}\sum_{\nu=1}^{r}\left|\frac{\nabla V_x(\rho,\hat q)\bigl[\Sigma_\nu(\rho,\hat q)
\bigr]}{V_x(\rho,\hat q)}\right|^2\geq 4x^2\Lambda_{\mathrm{red}}(\chi).
\label{eq:reduced-log-diffusion}
\end{equation}
The following result provides a unified reformulation of the robust reduced-observer stabilization results developed in \cite[Theorem~4.3]{liang2024model},
\cite[Theorem~4.14]{liang2021robust}, and \cite[Theorem~4.1]{liang2025exploring}.
Rather than reproducing any one of these results, we express their common recurrence--contraction mechanism under the notation and assumptions adopted here. Define $\mathbf d_{0}\big((\rho,\hat q),\mathcal I(\mathcal H_0)\times\{e_0\}\big):=d_0(\rho)+\|\hat q-e_0\|_1$.
\begin{theorem}
\label{thm:reduced-filter-GES}
Assume that  $\eta_\nu\in(0,1)$ for all $\nu\in[r]$, and Assumptions~\ref{ass:QND}, \ref{ass:one-sided-separation}, and \ref{ass:hamiltonian-nondecupling}--\ref{ass:contraction-dominance} 
hold. Then, 
\(
\mathcal I(\mathcal H_0)\times\{e_0\}
\)
is almost surely GES for  the coupled physical--observer system defined by~\eqref{eq:multi-homodyne-sme} and~\eqref{eq:reduced-QND-filter-W}. Moreover, for every
\(
(\rho_0,\hat q_0)\in\left(\mathcal S(\mathcal H)\times\boldsymbol\Delta_{d}^{\circ}\right)
\setminus\left(\mathcal I(\mathcal H_0)\times\{e_0\}\right),
\)
\begin{equation}
\limsup_{t\to\infty}\frac1t\log\mathbf d_{0}\big((\rho_t,\hat q_t),\mathcal I(\mathcal H_0)\times\{e_0\}\big)\leq-\Lambda_{\mathrm{red}}(\chi)+\frac12B_{\mathrm{red}}(\chi)<0,\quad a.s.
\label{eq:reduced-feedback-exponent}
\end{equation}
\end{theorem}

\begin{proof}
By Assumption~\ref{ass:contraction-dominance}, fix
\(
0<x<
1-\frac{B_{\mathrm{red}}(\chi)}
        {2\Lambda_{\mathrm{red}}(\chi)}.
\)
Then, due to \eqref{eq:reduced-LV-limsup}, there exist a neighborhood \(\mathcal U\) of \(\mathcal I(\mathcal H_0)\times\{e_0\}\) and \(c_x>0\) such that
\(
\mathscr L V_x(\rho,\hat q)
\leq
-c_xV_x(\rho,\hat q)
\)
on
\(
\mathcal U\setminus
(\mathcal I(\mathcal H_0)\times\{e_0\}).
\)
Therefore, a stopped-supermartingale argument~\cite[Theorem 4.2.2]{mao2007stochastic} implies local stability in probability.

Choose a neighborhood \(\mathcal U_1\)  such that
\(\overline{\mathcal U_1}\subset\mathcal U\) and, for some \(p>0\),
\[
\inf_{(\rho,\hat q)\in\mathcal U_1}
\mathbb P_{\rho,\hat q}
\left(
(\rho_t,\hat q_t)\in\mathcal U
\ \text{for all }t\geq0
\right)
\geq p.
\]
Define recursively the entrance times \(\sigma_k\) into
\(\mathcal U_1\) and the subsequent exit times \(\tau_k\) from
\(\mathcal U\). By similar arguments as in the proof of Theorem~\ref{thm:QND-GES}, we deduce $\mathbb P\left(\textstyle \bigcup_{k\geq1} E_k \right)=1$ where $E_k:=\{\sigma_k<\infty,\ \tau_k=\infty\}$.

Since \((\rho_t,\hat q_t) \in \mathcal S(\mathcal H)\times\boldsymbol\Delta_{d}^{\circ}\) for all $t\geq0$ almost surely, \(V_x(\rho_t,\hat q_t)>0\) almost surely. Fix \(\delta>0\). By
\eqref{eq:reduced-LV-limsup} and
\eqref{eq:reduced-log-diffusion}, \(\mathcal U\) can be chosen
sufficiently small so that
\[
\begin{aligned}
 &\frac{\mathscr L V_x(\rho,\hat q)}{V_x(\rho,\hat q)}
-\frac12
\sum_{\nu=1}^{r}
\left|
\frac{\nabla V_x(\rho,\hat q)
[\Sigma_\nu(\rho,\hat q)]}
{V_x(\rho,\hat q)}
\right|^2\\
&~~~~\leq
-x\left[
2\Lambda_{\mathrm{red}}(\chi)
-B_{\mathrm{red}}(\chi)
\right]+\delta, \qquad  \forall (\rho, \hat{q})\in \mathcal U\setminus (\mathcal I(\mathcal H_0)\times\{e_0\}).
\end{aligned}
\]
Applying It\^o's formula after the stopping time \(\sigma_k\), and using the strong law for continuous local
martingales yields, 
\[
\limsup_{t\to\infty}
\frac1t\log V_x(\rho_t,\hat q_t)
\leq
-x\left[
2\Lambda_{\mathrm{red}}(\chi)
-B_{\mathrm{red}}(\chi)
\right]+\delta, \quad a.s. \text{ on } E_k.
\]
Since
\(
\mathbb P(\bigcup_{k\geq1}E_k)=1
\)
and \(\delta>0\) is arbitrary, we obtain
\[
\limsup_{t\to\infty}
\frac1t\log V_x(\rho_t,\hat q_t)
\leq
-x\left[
2\Lambda_{\mathrm{red}}(\chi)
-B_{\mathrm{red}}(\chi)
\right],
\qquad\text{a.s.}
\]
Moreover, since
\(
\mathbf d_{0}\big((\rho_t,\hat q_t),\mathcal I(\mathcal H_0)\times\{e_0\}\big)^2
\leq c_{\mathrm{red}}V_x(\rho,\hat q)^{1/x}
\)
for some $c_{\mathrm{red}}>0$, we deduce
\[
\begin{aligned}
\limsup_{t\to\infty}
\frac1t
\log
\mathbf d_{0}\big((\rho_t,\hat q_t),\mathcal I(\mathcal H_0)\times\{e_0\}\big)
&\leq
\frac{1}{2x}
\limsup_{t\to\infty}
\frac1t\log V_x(\rho_t,\hat q_t)\\
&\leq
-\Lambda_{\mathrm{red}}(\chi)
+\frac12B_{\mathrm{red}}(\chi)
<0, \quad a.s.,
\end{aligned}
\]
where the last inequality follows from the contraction-dominance condition in Assumption~\ref{ass:contraction-dominance}. Together with local stability in probability, this proves the almost-sure GES of \(\mathcal I(\mathcal H_0)\times\{e_0\}\) and~\eqref{eq:reduced-feedback-exponent}.
\end{proof}

The preceding argument again separates the stabilization mechanism into
global recurrence and local exponential contraction. Its logical structure
can be summarized as follows:
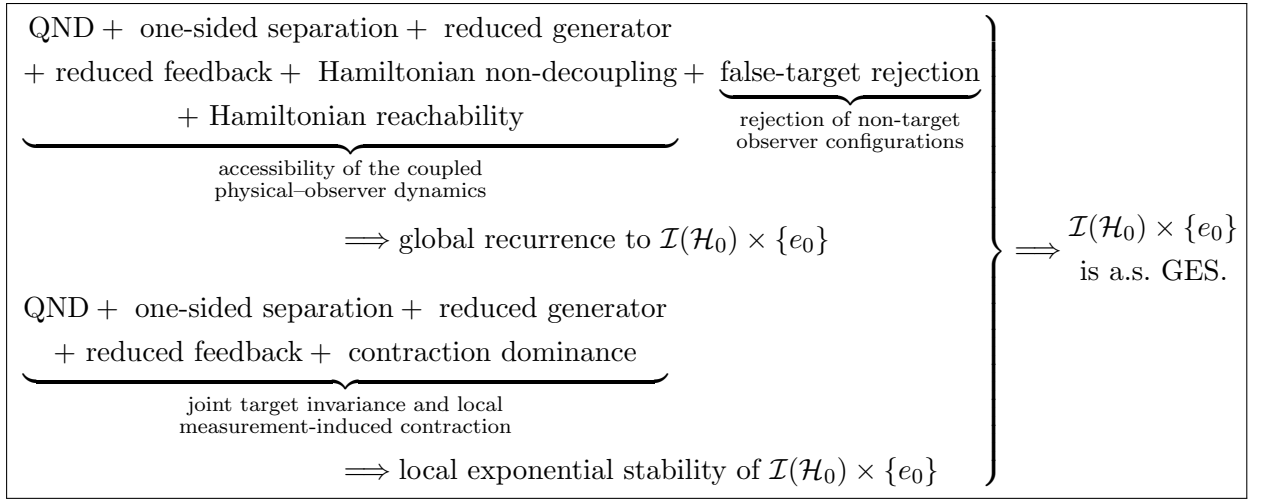
\begin{equation*}
\boxed{
\left.
\begin{aligned}
&
\underbrace{
\begin{gathered}
\text{QND}
+\text{ one-sided separation}
+\text{ reduced generator}\\
+\text{ reduced feedback}
+\text{ Hamiltonian non-decoupling}\\
+\text{ Hamiltonian reachability}
\end{gathered}
}_{\substack{\text{accessibility of the coupled}\\
\text{physical--observer dynamics}}}
+\ 
\underbrace{
\text{false-target rejection}
}_{\substack{\text{rejection of non-target}\\
\text{observer configurations}}}
\\[1mm]
&\hspace{42mm}
\Longrightarrow
\text{global recurrence to }
\mathcal I(\mathcal H_0)\times\{e_0\}
\\[3mm]
&
\underbrace{
\begin{gathered}
\text{QND}
+\text{ one-sided separation}
+\text{ reduced generator}\\
+\text{ reduced feedback}
+\text{ contraction dominance}
\end{gathered}
}_{\substack{\text{joint target invariance and local}\\
\text{measurement-induced contraction}}}
\\[1mm]
&\hspace{42mm}
\Longrightarrow
\text{local exponential stability of }
\mathcal I(\mathcal H_0)\times\{e_0\}
\end{aligned}
\right\}
\Longrightarrow
\begin{gathered}
\mathcal I(\mathcal H_0)\times\{e_0\}\\
\text{is a.s.\ GES.}
\end{gathered}
}
\end{equation*}

\begin{remark}
The constants quantify different parts of the stabilization mechanism.
The physical constant \(\Lambda_0\) controls the local contraction of the
true target-population error, the pairwise exponent
\(\Lambda_{i\mid j}(\chi)\) controls rejection of observer sector \(i\) when
the physical trajectory approaches sector \(j\), and
\(\Lambda_{\mathrm{red}}(\chi)\) controls the combined
physical--observer contraction.
Choosing \(a=1\) in~\eqref{eq:reduced-diffusion-constant} gives
\begin{equation*}
 \Lambda_{\mathrm{red}}(\chi) \leq \inf_{\pi\in\boldsymbol\Delta_{d-1}} \sum_{\nu}
\hat\iota_\nu^2\chi_\nu^2A_\nu(\pi)^2 = \Lambda_0.
\end{equation*}
In the nominal case, we obtain
\(
B_{\mathrm{red}}(1)=0.
\)
Moreover, due to the linearity of \(A_\nu\), we deduce $\Lambda_{\mathrm{red}}(1)=\Lambda_0.$
Thus, the nominal reduced-order observer recovers the state-feedback exponent
\(-\Lambda_0\).
Finally, since \(A_\nu(e_n)=\delta_{\nu,n}\), for every \(n\in[d]\), we have
\begin{equation*}
\Lambda_0\leq\sum_{\nu}\iota_\nu^2\delta_{\nu,n}^2=\frac12\Lambda_{0\mid n}(1).
\end{equation*}
The inequality may be strict since the minimum defining \(\Lambda_0\) can
be attained at a nontrivial convex combination of non-target sectors.
\end{remark}

\section{Implementation challenges and future directions}
\label{sec:implementation_future}

The preceding sections have identified structural conditions under which
measurement-based feedback stabilizes continuously monitored quantum
systems. These results are commonly derived under idealized assumptions:
the conditional state can be propagated exactly, the measurement record is
available continuously, feedback is applied without delay, and the
closed-loop dynamics admits a finite-dimensional Markovian description.
Thus, translating these results into experimentally viable control architectures requires more than efficient numerical integration. It calls for
control-oriented model reduction, sampled-data and finite-bandwidth
analysis, online identification, and systematic extensions to hybrid and
non-Markovian settings.

Stabilization also captures only one aspect of closed-loop performance. A
feedback law may guarantee convergence without optimizing the preparation
time, control effort, leakage, robustness margin, or sensitivity to
measurement noise. Moreover, when the available actuation cannot restore
invariance of the desired state or subspace, convergence to a fixed target
is no longer the appropriate asymptotic objective. These limitations
motivate the implementation issues and research directions discussed below.

\subsection{Scalable estimation and low-complexity feedback}
\label{subsec:future_reduced_filtering}

Scalability arises at several levels in measurement-based quantum feedback.
For a system on an \(N\)-dimensional Hilbert space, a general density operator
contains \(N^2-1\) independent real parameters, so propagating the full
conditional state can become costly at high measurement-update rates.
For a plant composed of many interacting subsystems, the joint Hilbert-space
dimension introduces an additional exponential dependence on the number of
subsystems. Even when estimation is tractable, the controller itself must
operate within finite computational, bandwidth, and latency constraints.
These considerations motivate complementary reductions of the estimator,
the many-body model, and the controller realization.

\paragraph{Reduced estimation.}
The reduced QND filter introduced in Subsection~\ref{sec:observer-sme} illustrates a
general principle: full-state reconstruction is unnecessary
when the feedback law depends only on a smaller set of conditional variables.
Extending this principle beyond the QND setting is a central problem for
scalable quantum feedback.

Two complementary approaches are available. Projection filters approximate
the conditional state on a prescribed finite-dimensional manifold, typically
using information-geometric or related projection methods
\cite{van2005quantum,nielsen2009quantum,tezak2017low,gao2019design,
gao2020improved,amini2025feedback}. Exact reduction instead seeks a
lower-dimensional Belavkin equation that reproduces the conditional
expectations of selected observables, using observability, minimal-realization
methods, and noncommutative conditional expectations
\cite{grigoletto2025quantum}. However, for feedback design, approximation or
realization accuracy alone is insufficient. The reduced state must retain the
information and structural properties required by the controller, such as
target invariance, distinguishability of relevant sectors, and the Lyapunov
or recurrence properties underlying closed-loop stability. Developing
control-oriented reduction criteria that quantify this trade-off remains an
important problem.

\paragraph{Mean-field reduction for large monitored ensembles.}
A different scalability problem arises when the plant consists of a large
number of interacting and continuously monitored subsystems. Let
\(\mathfrak h\simeq\mathbb C^n\) be the one-particle Hilbert space and
$
\mathcal H_M:=\mathfrak h^{\otimes M}
$
the Hilbert space of an \(M\)-particle system. Even for small \(n\), direct
propagation of a general conditional state on \(\mathcal H_M\) becomes
intractable as \(M\) increases.

For weak mean-field interactions, exchangeable initial data, and suitable
local measurement structures, the many-particle conditional dynamics may
admit a nonlinear one-particle limit. Schematically, the limiting conditional
state satisfies an equation of the form
\begin{equation*}
d\varrho_t
=
\mathcal L_{u_t,\bar\varrho_t}(\varrho_t)\,dt
+
\sum_{\nu=1}^{r}
\iota_\nu
\mathcal G_{L_\nu}(\varrho_t)\,dW_{\nu,t},
\qquad
\bar\varrho_t:=\mathbb E[\varrho_t],
\end{equation*}
where the dependence on \(\bar\varrho_t\) represents the effective
mean-field interaction. The conditional state \(\varrho_t\) remains
stochastic because of the local measurement record, while the interaction
with the remaining particles enters through the averaged one-particle state.
The corresponding propagation-of-chaos property takes the form
\begin{equation*}
\sup_{0\leq t\leq T}
\mathbb E
\left[
\left\|
\rho_t^{M:k}
-
\varrho_t^{1}\otimes\cdots\otimes\varrho_t^{k}
\right\|_1
\right]
\longrightarrow0,
\qquad M\to\infty,
\end{equation*}
for every fixed \(k\) and \(T>0\), where \(\rho_t^{M:k}\) is the
\(k\)-particle marginal and
\(\varrho^1,\ldots,\varrho^k\) are independent copies of the limiting
process. This reduction is complementary to the projection and realization
methods above: those methods reduce the number of variables used to represent
one conditional state, whereas propagation of chaos reduces the number of
interacting subsystems that must be represented jointly.

Mean-field Belavkin equations and propagation-of-chaos results have been
developed for finite-dimensional monitored systems, with extensions to
heterogeneous graphon interactions, infinite-dimensional models, and
mixed-state dynamics under inefficient measurements
\cite{chalal2023meanfieldBelavkin,amini2025graphon,
deBouard2026meanfield,guo2026propagation}; see also
\cite{kolokoltsov2026quantumFiltering} for a recent mathematical review of
quantum filtering, propagation of chaos, and their applications to feedback
control and mean-field games. These results suggest decentralized
feedback laws based on a local conditional state together with a mean-field
statistic. However, their use for feedback synthesis requires more than
fixed-policy propagation of chaos. In particular, convergence of the
finite-particle and limiting dynamics for each prescribed feedback law does
not by itself imply convergence of the corresponding stability properties or
optimal-control problems. Thus, a control-oriented theory should seek
estimates that are uniform over relevant feedback classes, together with
procedures for transferring stabilizing or nearly optimal policies between
the limiting and finite-particle models. Uniform-in-time estimates under
stabilizing feedback and conditional propagation of chaos under common or
collective measurement noise remain particularly important open problems.

\paragraph{Low-complexity controller realizations.}
Reduction of the information state does not by itself guarantee an
implementable feedback law. Thus, a complementary route is to reduce the
complexity of the controller itself. To describe both deterministic and
stochastic actuation, let \(v_t\) be a classical semimartingale control
signal. The corresponding Hamiltonian action is naturally interpreted in
Stratonovich form $dU_t=-\mathrm{i}H_1U_t\circ dv_t$.
When \(v_t\) has finite variation, $dv_t=u_tdt$
and this reduces to the usual Hamiltonian control with amplitude \(u_t\).
When \(v_t\) has nonzero quadratic variation, conversion to It\^o form
introduces the corresponding quadratic-variation correction.
A general finite-dimensional measurement-driven controller may be represented as
\begin{align}
{d}z_t&=f_c(z_t){d}t+ \sum_{\nu}g_{c,\nu}(z_t){d}Y_{\nu,t},\label{eq:general_controller_state}\\
{d}v_t&=a_c(z_t){d}t+  \sum_{\nu}b_{c,\nu}(z_t){d}Y_{\nu,t}+c_c(z_t){d}B_t,\label{eq:general_control_increment}
\end{align}
where \(z_t\) is the controller state and \(B_t\) is a Wiener process
independent of the measurement noises. Equations
\eqref{eq:general_controller_state} and
\eqref{eq:general_control_increment} specify the classical controller, while
its interconnection with the plant determines the complete closed-loop SME.

In filtering-based Bayesian feedback, \(z_t=\hat\rho_t\), or
\(z_t=\hat q_t\) for a reduced observer, and the actuation typically has
finite variation, $dv_t=a_c(z_t)dt.$
The principal online cost is then the propagation of the conditional state
or its reduced representation. In analog-filter feedback, \(z_t\) is instead
the internal state of a physical signal-processing circuit. For example, a
first-order low-pass controller may be modeled as
\[
\textstyle dz_t
=
-\omega_c z_t dt+\omega_c dY_t,
\qquad
dv_t
=
\kappa_c(z_t) dt.
\]
No conditional density operator is reconstructed, reducing the online
computational burden at the cost of finite bandwidth and additional memory
in the feedback loop. The quantum-enhanced optical phase-tracking experiment
of Yonezawa \emph{et al.}~\cite{yonezawa2012quantum} provides a representative
implementation in which a homodyne record is processed by a real-time tracking
filter before feedback is applied to the local-oscillator phase.
Wiseman--Milburn Markovian feedback~\cite{wiseman2009quantum} corresponds to
the memoryless direct-feedthrough case
$
dv_t=fdt+\sum_{\nu}\sigma_\nu dY_{\nu,t},
$
with no controller state \(z_t\). It may be viewed as an idealized
negligible-delay limit of electronic measurement feedback. Noise-assisted
feedback instead introduces an additional stochastic actuation,
$
dv_t=\sigma(z_t)dB_t,
$
with a gain determined by the conditional state or a reduced estimate
\cite{cardona2020exponential}. The resulting control diffusion can remove
undesired invariant configurations while allowing the actuation to vanish
near the target. Bayesian, analog, Markovian, and noise-assisted feedback
therefore represent distinct realizations of
\eqref{eq:general_controller_state} and \eqref{eq:general_control_increment}
and may also be combined.

Therefore, the central problem is control-relevant reduction: determining the smallest information state and simplest controller realization that retain the stability, robustness, and performance properties required by the closed-loop objective. This requires linking estimator and model reduction errors, controller complexity, sampling, bandwidth, and delay directly to closed-loop guarantees.


\subsection{Sampling, delay, and experimental constraints}
\label{subsec:future_sampling_delay}

Continuous-time SMEs are formulated in terms of an ideal measurement record
\(\{Y_{\nu,t}\}_{t\geq0}\). However, in real experiments, the detector output is
amplified, filtered, and digitized over sampling intervals of duration
\(\Delta>0\). The datum acquired from the \(\nu\)-th measurement channel
during the \(k\)-th interval may be represented by
\[
\begin{aligned}
I_{\nu,k}:=\int_{(k-1)\Delta}^{k\Delta}I_{\nu,t}{d}t=\int_{(k-1)\Delta}^{k\Delta}{d}Y_{\nu,t}=Y_{\nu,k\Delta}-Y_{\nu,(k-1)\Delta},\quad k\geq1,
\end{aligned}
\]
where
\(
I_{\nu,t}:={d}Y_{\nu,t}/{d}t
\)
is understood formally for a diffusive measurement record. Thus, the
experimentally accessible data are the discrete sequence
\(\{I_{\nu,k}\}_{k\geq1}\) rather than the complete continuous-time
trajectory.

The feedback input is updated at discrete times and is
typically held constant or interpolated between successive samples. Therefore, the
implemented closed loop is more accurately described as a sampled-data stochastic system~\cite{rouchon2022tutorial}.
Stability properties established for the continuous-time SME do not
automatically carry over to its digital implementation. A systematic theory
should quantify how the sampling period, measurement strength, feedback
gain, and nominal convergence rate jointly determine closed-loop behavior.
In particular, one seeks conditions of the form $\Delta<\Delta_{\max}$ preserving the stability or practical stability, together with quantitative estimates of the degradation in convergence rate as \(\Delta\) increases.

Finite sampling also changes the conditional state relevant for estimation
and feedback. Let
\(
\mathcal F_k^{I}:=\sigma(I_{\nu,j}:1\leq\nu\leq r,1\leq j\leq k)
\)
denote the information generated by the digitized observations up to time
\(k\Delta\). The corresponding Bayesian information state is
\(
\rho^{\Delta}_k:=\mathbb E[\rho_{k\Delta}|\mathcal F_k^{I}].
\)
In general, \(\rho^{\Delta}_k\) cannot be recovered exactly by simply
inserting the finite-bin observations into a first-order discretization of
the continuous-time SME. Exact finite-bin quantum instruments and systematic
higher-order approximations have recently been developed for time-averaged
continuous measurements \cite{guilmin2025time,wonglakhon2026quantum}.
Extending these constructions to feedback stabilization and closed-loop
performance analysis remains an important open problem.

Feedback delay may arise without modifying the underlying Markovian dynamics
of the quantum plant. Such an \emph{information-path delay} results from measurement acquisition, signal processing, state estimation, communication, control computation, or actuation. 
It may be represented schematically as $u_t=u(\hat\rho_{t-\tau_c})$ or $u_t=u(\hat q_{t-\tau_c}).$ Although the plant may remain Markovian, the closed-loop dynamics become history dependent since the controller acts on delayed information. For
particular continuously monitored systems, delay-dependent stability has
been studied for feedback laws subject to estimation or computation delays
\cite{kashima2009control,ge2012quantum}. In continuous time, fixed delays
naturally lead to dynamics on a history space, while the
Lyapunov--Krasovskii and Razumikhin functionals used in stochastic delay
analysis \cite{mao2007stochastic} are difficult to construct globally on the quantum state space. Developing verifiable conditions that combine global recurrence with local delay-robust stability is an important research direction.

The photon-box experiment~\cite{sayrin2011realtime} provides a physically instructive example. After
interacting with the cavity, the Rydberg atoms require a finite time to reach
the detector. Then, the most recent measurement outcomes are unavailable
when the next control action is computed. This effect is modeled as a known
\(d\)-step delay in the discrete-time quantum trajectory and compensated by
augmenting the filter with the pending control actions and employing a
stochastic Smith predictor
\cite{amini2012stabilization,amini2013feedback}.
Although the delay originates from propagation of the probe atoms, it enters
the model as a delayed observation and control channel rather than as memory
in the cavity dynamics. Propagation-induced physical memory is introduced
separately in Subsection~\ref{subsec:future_nonmarkovian}.

Finite detector bandwidth introduces a further departure from the ideal
instantaneous-measurement model. By augmenting the state with detector
variables, one may derive a joint quantum--classical filtering equation that
captures detector dynamics, signal processing, and feedback within a unified
framework \cite{annby2022quantum}.
Such models provide a natural basis for determining when the ideal SME
remains an accurate approximation and for quantifying how finite bandwidth
affects estimation and closed-loop performance.

Other experimental constraints include bounded control amplitudes, finite
slew rates, detector inefficiency, dark counts, amplifier noise, calibration
errors, uncertainty in the measurement phase, data loss, and leakage outside
a truncated Hilbert space. These effects motivate feedback laws designed
under explicit hardware constraints, rather than unconstrained controls
followed by ad hoc saturation. Measurement strength should likewise be
treated as a design variable: stronger monitoring provides more information,
but also increases measurement back-action and generally requires larger
detector and controller bandwidths. These considerations point toward a
co-design framework in which the measurement channel, estimator, controller,
sampling architecture, and experimental hardware are treated jointly.

\subsection{Adaptation, uncertainty, and hybrid operation}
\label{subsec:future_adaptive_hybrid}

The observer-based results reviewed above can handle fixed parameter
mismatches; however, they do not generally account for unknown parameters or variations during operation. In practice, coupling strengths, detunings, measurement efficiencies, decoherence rates, and actuator gains may be imperfectly known or subject to slow drift. A natural extension is to combine the quantum filter with an online parameter estimator. Schematically, an adaptive
observer may be written as
\[
\begin{aligned}
 d\hat\rho_t&=\mathcal L_{\hat\theta_t,u_t}\bigl(\hat\rho_t\bigr) dt+ \sum_{\nu}\mathcal G_{\nu,\hat\theta_t}\bigl(\hat\rho_t\bigr)d\widehat W_{\nu,t},
\\
d\hat\theta_t&=a\bigl(\hat\rho_t,\hat\theta_t,u_t\bigr) dt+ \sum_{\nu}b_\nu\bigl(\hat\rho_t,\hat\theta_t,u_t\bigr)d\widehat W_{\nu,t},
\end{aligned}
\]
where
\(
d\widehat W_{\nu,t}=dY_{\nu,t}-h_{\nu,\hat\theta_t}\bigl(\hat\rho_t\bigr)dt
\)
is the innovation process associated with the current parameter estimate.
Before the estimated model converges to the true one, this process may not
be a Wiener process under the physical probability measure. Representative
approaches include Bayesian parameter augmentation, banks of candidate
filters, online maximum-likelihood and score-based methods, and adaptive
tuning laws \cite{enami2021proposal,clausen2024online}.

The main difficulty is the coupling between identification and control. Rapid
stabilization may drive the system into a regime in which the measurement
record contains little information about the unknown parameters, whereas
additional excitation introduced for identification may degrade regulation
and increase measurement back-action. For constant unknown parameters, one
seeks closed-loop identifiability, consistency of \(\widehat\theta_t\), and joint convergence of \((\rho_t,\hat\rho_t,\hat\theta_t)\). 
For slowly varying parameters, quantitative tracking bounds are more
appropriate. Further questions concern preservation of positivity and
normalization, confinement of parameter estimates to physically admissible
sets, and the use of finite-time confidence regions in robust feedback
design.

Abrupt faults, discrete measurement events, and changes of operating regime
require a hybrid description rather than a continuously varying parameter
model. Homodyne and heterodyne detection generate diffusive records, whereas
photon counting produces jump processes; simultaneous or switched
measurements therefore lead naturally to jump--diffusion SMEs
\cite{bouten2007introduction}. Classical modes may additionally represent
detector configurations, actuator faults, finite-state environments, or
operating regimes. The resulting quantum--classical process may remain
Markovian after these variables are included in the augmented state
\cite{barchielli2024hybrid,barchielli2024markovian}.

A related hybrid structure arises from deliberate switching among feedback
laws. Even for a purely diffusive plant, the closed loop becomes hybrid when
a discrete controller mode selects among different Hamiltonian or dissipative
actions. The switching constructions reviewed earlier show how complementary
control mechanisms can be combined, and how mode-wise invariance requirements may be relaxed while avoiding chattering or Zeno behavior through suitable modulation or hysteresis \cite{liang2024switching}.
Extending these methods to reduced observers, jump--diffusion dynamics,
delayed feedback, and uncertain switching mechanisms remains an important
direction.

Hybrid models also provide a natural framework for fault-tolerant feedback.
Quantum--classical filters may be used to estimate both the conditional
quantum state and an active fault mode
\cite{gao2016fault,wang2016fault,yu2020hybrid}.
The controller must then detect the mode transition, reconfigure the observer
and feedback law, and recover an appropriate stability or practical-stability
property. A systematic theory should combine mode estimation, switching,
robustness to missed or spurious detections, and nonlinear SME stabilization.
When a fault or mode transition destroys invariance of the original target,
practical stability and long-time statistical objectives become more
appropriate, as discussed in Subsection~\ref{subsec:future_noninvariant_target}.

\subsection{Non-Markovian quantum feedback}
\label{subsec:future_nonmarkovian}

Unlike an information-path delay, which postpones the availability or
application of feedback information, \textit{propagation delay} may introduce memory directly into the physical dynamics. When an output field travels through an optical path or waveguide and subsequently interacts again with the plant, it
carries quantum information from earlier system--field interactions. The
reduced plant dynamics is generally non-Markovian~\cite{barchielli2012quantum,grimsmo2015time,
pichler2016photonic}.

The Markovian SME approximation may fail for structured
reservoirs, strongly coupled ancillary systems, colored environments, and
coherent feedback networks with non-negligible propagation times. Certain
memory effects can be represented through generalized Lindblad equations and their jump--diffusion unravellings
\cite{barchielli2010jump}.
Alternatively, the principal system may be embedded into a larger Markovian
model containing ancillary modes, pseudomodes, or other degrees of freedom that encode memory effects. Markovian quantum-network representations, often formulated within the SLH framework~\cite{combes2017slh}, characterize each component by a triple $(S,L,H)$ describing scattering, system-field coupling, and internal Hamiltonian, and provide modular rules for network interconnection. These representations offer a systematic way to construct augmented Markovian models, from which a quantum filter can be derived for the enlarged conditional state
\cite{gough2012single,xue2019modeling,nurdin2025physical}.
Markovian embedding shifts the control problem to a larger state space.
Stability properties of the augmented filter must be translated into
corresponding statements for the principal system marginal, while the memory
subsystem may be only partially observed. The enlarged conditional state may
also be more expensive to propagate than the original filter.

Important open problems include control-oriented reduction of monitored
Markovian embeddings, observer stability under incomplete access to memory
degrees of freedom, and feedback design based only on principal-system
measurements. It is also desirable to derive stability and robustness
conditions directly in terms of memory kernels, environmental spectra, or
delayed input--output relations, without requiring an explicit
high-dimensional embedding. A systematic comparison between
measurement-based and coherent feedback in the presence of propagation
memory remains largely undeveloped.

\subsection{Non-invariant targets, stochastic optimal control, and learning}
\label{subsec:future_noninvariant_target}
\label{subsec:future_optimal_learning}

Exact asymptotic stabilization requires invariance of the target under the
admissible closed-loop dynamics. When an uncontrolled channel destroys
nominal invariance, the first question is whether the available control
resources can restore it. Switching  control may sometimes resolve this issue when individual closed-loop modes do not preserve the
target \cite{liang2024switching}.
However, in other settings, energy relaxation, persistent leakage, or
actuator limitations make exact invariance unattainable. Convergence to a
fixed state or subspace is then no longer an appropriate asymptotic
objective.

A first alternative is \emph{practical stabilization}, in which the
trajectory approaches and remains near the nominal target, with a residual
error determined by the disturbance magnitude and the available control
authority. Relevant objectives include input-to-state stability,
convergence to an invariant neighborhood, and quantitative estimates of the
asymptotic target error \cite{liang2025exploring}.
For persistent stochastic disturbances, the long-time behavior is more
naturally characterized through invariant probability measures. Fundamental
questions concern existence and uniqueness of an invariant measure,
convergence in law, mixing rates, metastability, and concentration of the
stationary distribution near the desired subspace~\cite{meyn2012markov}. Results of this type have
been established for classes of uncontrolled quantum trajectories under
ergodicity and purification assumptions \cite{benoist2021invariant}. However,  the extension to nonlinear feedback-controlled SMEs remains open.

If the closed-loop process under a feedback law \(u\) admits a unique
invariant measure \(\mu_u\), its stationary target error may be quantified by
\[
\mathcal J_{\mathrm{stat}}(u)
:=
\int_{\mathcal S(\mathcal H)}
d_S(\rho)\,
\mu_u(\mathrm{d}\rho).
\]
More general stationary criteria may include leakage, energy, or average
control effort. The design objective is not pointwise convergence, but
concentration of the long-time distribution near the desired target. This
formulation leads naturally to ergodic stochastic control~\cite{benaim2022markov}.

When exact stabilization is feasible, stability alone does not determine
the quality of the transient response. Stabilizing controllers may differ
substantially in preparation time, control effort, leakage, robustness
margin, and sensitivity to measurement noise. For a continuously monitored
system, a finite-horizon stochastic optimal-control problem~\cite{yong1999stochastic} may be written as
\begin{equation*}
J_{t,\rho}(u):=\mathbb E_{t,\rho}^{u}\left[\Phi(\rho_T)+\int_t^T\ell(\rho_s,u_s){d}s
\right],
\end{equation*}
where \(\Phi\) is a terminal cost and \(\ell\) may penalize infidelity,
control effort, leakage, or violations of experimental constraints. Other
relevant criteria include minimum expected hitting time, risk-sensitive
costs, long-time average performance, and probabilities of leaving a
prescribed safe set.

Define
\(
V(t,\rho):=\inf_{u\in\mathcal U_{t,T}}J_{t,\rho}(u).
\)
Dynamic programming formally yields
\begin{equation*}
 -\partial_t V(t,\rho)=\inf_{u\in\mathcal U}\left\{\ell(\rho,u)+\mathcal A^{u}V(t,\rho)\right\},\quad V(T,\rho)=\Phi(\rho),
\end{equation*}
where \(\mathcal A^{u}\) is the infinitesimal generator of the controlled SME
\cite{bouten2005bellman,gough2005hamilton}.
The conditional density operator serves as an information state, but the
resulting Hamilton--Jacobi--Bellman equation is posed on a state-constrained
compact convex set and is generally degenerate. Direct solution is restricted to low-dimensional systems, making reduced filters, symmetry
reduction, approximate dynamic programming, and model-predictive control
particularly relevant.

The measurement channel may itself be included in the optimization. If
\(\alpha_t\) denotes a measurement setting, such as the local-oscillator
phase, detection basis, or measurement strength, an adaptive measurement may
be represented schematically as
\(
\alpha_t = \pi_t\bigl(\mathcal F_t^Y\bigr),
\)
and then future observations depend on the previously accumulated record.
Wiseman's adaptive phase-measurement scheme provides a canonical example in
which the local-oscillator phase is updated in real time to improve the
information extracted from subsequent homodyne measurements
\cite{wiseman1995adaptive,armen2002adaptive}.
Thus, the adaptive measurement can be naturally viewed as part of a joint
measurement--estimation--control problem.

Model-based numerical optimization provides a complementary route when
closed-form feedback design or dynamic programming becomes intractable.
Gradient-based methods such as GRAPE, ensemble-control formulations for
systems subject to parameter dispersion, and pseudospectral schemes for
open-system optimal control provide complementary approaches to
high-fidelity quantum-control design
~\cite{khaneja2005optimal,li2009pseudospectral}.
Reinforcement learning and related data-driven approaches can instead use
the conditional state, a reduced observer, or a finite-dimensional summary
of the measurement history as the policy input
~\cite{dong2023learning,guatto2024improving,ma2025machine,
song2025fast}.
The principal challenge is certification: a policy that performs well in
simulation may fail to preserve target invariance, satisfy hardware
constraints, or remain robust to model mismatch and operating conditions
outside the training distribution.

The broader objective is a joint design of measurement,
estimation, and actuation that balances information acquisition,
measurement back-action, transient performance, robustness, and hardware
limitations. Developing scalable optimization and learning methods with
verifiable closed-loop guarantees remains a central challenge.

\subsection{Outlook}
\label{subsec:future_outlook}

The QND setting provides a particularly transparent methodological benchmark
for measurement-based quantum feedback. The relevant information state,
reduced filtering dynamics, and stabilization mechanism can be expressed in
terms of sector populations, while convergence can be characterized through
explicit exponential rates. More broadly, the QND analysis illustrates how
measurement structure may be exploited to obtain low-dimensional estimators,
implementable feedback laws, and quantitative closed-loop guarantees.

The central challenge is to retain this control-theoretic transparency when
the observable structure is noncommutative, the information state is
high-dimensional, or the implemented dynamics is sampled, delayed, uncertain,
hybrid, or non-Markovian. Switching feedback and adaptive measurement further
show that an effective closed-loop architecture do not need consisting of a single
smooth state-feedback law: discrete controller modes, reduced observers, and
online modification of the observation channel may all play essential roles.

Progress in these directions also requires common standards for numerical and
experimental evaluation. Nominal sample trajectories alone provide limited
evidence for stochastic closed-loop performance. Comparisons should report convergence or hitting-time statistics, control effort, computational
cost per measurement update, memory requirements, and robustness across
trajectory ensembles and model uncertainties. Whenever applicable, reduced
or learning-based controllers should be compared with open-loop, direct
Markovian-feedback, and full-filter baselines. Hardware-in-the-loop validation
is particularly important since estimator complexity, sampling rate,
bandwidth, and feedback delay are coupled. Reproducible studies should also
specify the integration method, sampling period or numerical tolerance,
control constraints, uncertainty model, and statistical confidence of the
reported results.

For bosonic cavities, optomechanical systems, and collective atomic
ensembles, numerical validation must additionally account for
Hilbert-space truncation and model-reduction errors. Stability and performance
estimates that remain uniform with respect to the truncation dimension would
provide an important link between finite-dimensional SME theory and
continuous-variable experimental platforms.

Therefore, further progress is unlikely to arise from a universal feedback
formula alone. A more realistic objective is to develop a systematic
methodology for measurement-based feedback. This requires determining what
information is needed for control, designing estimators and controllers
compatible with the available hardware, and establishing performance
guarantees appropriate to the control objective, including asymptotic
stabilization, practical stability, and long-time statistical performance.
Integrating quantum filtering, nonlinear stochastic control, model reduction,
optimization, reproducible computation, and experimental implementation
within such a methodology remains a central challenge for measurement-based
quantum feedback.




\bibliographystyle{plain}
\bibliography{cas-refs}

\end{document}